\newif\ifqip
\qipfalse

\newif\ifeprint
\eprinttrue

\newif\ifsubmit
\submittrue

\documentclass[11pt]{article}
\usepackage[pdftex,left=1in,top=1in,bottom=1in,right=1in]{geometry}

\usepackage{amsmath}
\usepackage{amssymb}
\usepackage{amsthm}
\usepackage{mathtools}
\usepackage{xcolor}
\usepackage{hyperref}
\usepackage{cleveref}
\usepackage{thm-restate}
\usepackage{xargs}
\usepackage{enumitem}
\usepackage{ytableau}
\usepackage{array,multirow,boldline}
\usepackage{float}
\usepackage{tikz-cd}
\usepackage{dsfont}
\usepackage{bbm}
\usepackage{braket}
\usepackage{xspace}
\usepackage{algorithm}
\usepackage{algpseudocode}
\usepackage[most]{tcolorbox}
\usepackage{thmtools}
\usepackage{thm-restate}

\usetikzlibrary{quantikz2}
\usetikzlibrary{shapes,arrows.meta,positioning}
\usetikzlibrary{calc,decorations.pathreplacing}
\usetikzlibrary{arrows.meta}

\newtcolorbox{algbox}[1]{
  enhanced,
  colback=white,
  colframe=black,
  coltitle=black,
  title={#1},
  attach boxed title to top left={xshift=3mm,yshift=-2mm},
  boxed title style={
    colback=white,
    colframe=black,
    boxrule=0.6pt,
    sharp corners
  },
  sharp corners,
  boxrule=0.6pt,
  left=2mm,
  right=2mm,
  top=4mm,
  bottom=2mm
}

\ifeprint
\usepackage{tocloft}

\fi

\allowdisplaybreaks

\theoremstyle{definition}
\newtheorem{theorem}{Theorem}
\newtheorem{definition}{Definition}

\newtheorem{lemma}{Lemma}

\newtheorem{conjecture}{Conjecture}

\newtheorem{remark}{Remark}

\newtheorem{construction}{Construction}

\newcommand{\Tr}{\operatorname{Tr}}

\newcommand{\BQP}{\mathsf{BQP}}
\newcommand{\QMA}{\mathsf{QMA}}
\newcommand{\NP}{\mathsf{NP}}

\newcommand{\QFT}{\algo{QFT}}

\newcommand{\cD}{\mathcal{D}}

\newcommand{\cM}{\mathcal{M}}

\newcommand{\cO}{\mathcal{O}}

\def\makeuppercase#1{
\expandafter\newcommand\csname tl#1\endcsname{\widetilde{#1}}
}

\def\makelowercase#1{
\expandafter\newcommand\csname tl#1\endcsname{\widetilde{#1}}
}

\newcommand{\Is}{\mathcal{I}}

\newcommand{\secp}{\lambda}

\newcommand{\A}{\entity{A}}
\newcommand{\B}{\entity{B}}

\newcommand*{\algo}[1]{\ensuremath{\mathsf{#1}}}

\newcommand*{\entity}[1]{\mathcal{#1}}
\newcommand{\compclass}[1]{\operatorname{#1}}

\newenvironment{boxfig}[2]{\begin{figure}[#1]\fbox{\begin{minipage}{0.97\linewidth}
                        \vspace{0.2em}
                        \makebox[0.025\linewidth]{}
                        \begin{minipage}{0.95\linewidth}
            {{
                        #2 }}
                        \end{minipage}
                        \vspace{0.2em}
                        \end{minipage}}}{\end{figure}}

\newcommand{\bit}{\{0,1\}}

\newcommand{\KeyGen}{\algo{KeyGen}}

\newcommand{\Sign}{\algo{Sign}}

\newcommand{\Ver}{\algo{Ver}}

\newcommand{\cBQPqpoly}{\mathsf{BQP}/\compclass{clonableqpoly}}
\newcommand{\cBQPqpolyO}{\mathsf{BQP}^{\mathcal{O}}/\compclass{clonableqpoly}}
\newcommand{\BQPqpoly}{\mathsf{BQP}/\compclass{qpoly}}

\newcommand{\BQPpoly}{\mathsf{BQP}/\compclass{poly}}
\newcommand{\BQPpolyO}{\mathsf{BQP}^{\mathcal{O}}/\compclass{poly}}

\makeatletter
\DeclareRobustCommand
  \myvdots{\vbox{\baselineskip4\p@ \lineskiplimit\z@
    \hbox{.}\hbox{.}\hbox{.}}}
\makeatother

\renewcommandx{\Pr}[3][1=,2=\big]{\mathop{\operatorname{Pr}}_{#1} #2 [ #3 #2 ]}
\newcommandx{\E}[3][1=,2=\big]{\mathop{\mathbb{E}}_{#1} #2 [ #3 #2 ]}

\newcommand{\negl}{\operatorname{negl}}
\newcommand{\bin}{\operatorname{\mathsf{BinEnc}}}

\newcommand{\bA}{\mathbf{A}}
\newcommand{\bB}{\mathbf{B}}

\newcommand{\bI}{\mathbf{I}}

\newcommand{\bO}{\mathbf{O}}

\newcommand{\bX}{\mathbf{X}}

\newcommand{\bWa}{\mathbf{W_1}}
\newcommand{\bWb}{\mathbf{W_2}}

\newcommand{\cQMA}{\mathsf{ClonableQMA}}

\newcommand{\QCMA}{\mathsf{QCMA}}

\newcommand{\Good}{\mathsf{Good}}
\newcommand{\YZ}{\mathsf{YZ}}
\newcommand{\BiasedYZ}{\mathsf{BiasedYZ}}
\newcommand{\Bias}{\mathsf{Bias}_{q,1/\lambda^{4},\Sigma}}
\newcommand{\advstate}{\ket{\mathsf{adv}_H}}

\newcommand{\Clone}{\mathsf{Clone}}
\newcommand{\Hash}{\mathsf{Hash}}

\newcommand{\Guesser}{\mathsf{Guesser}}

\newcommand{\email}[1]{\href{mailto:#1}{#1}}

\newcommand{\smallset}{\mathsf{SmallSet}_H}

\providecommand{\operationname}[1]{\operatorname{#1}}
\providecommand{\ket}[1]{\lvert #1\rangle}
\providecommand{\bra}[1]{\langle #1\rvert}
\providecommand{\QCMA}{\operationname{QCMA}}
\providecommand{\QMA}{\operationname{QMA}}

\providecommand{\QPT}{\operationname{QPT}}
\providecommand{\negl}{\operationname{negl}}
\providecommand{\Good}{\mathsf{Good}}
\providecommand{\wBias}{\mathsf{Bias}}

\providecommand{\QFireMini}{\mathsf{QFireMini}}

\providecommand{\DS}{\mathsf{DS}}

\providecommand{\KeyGen}{\mathsf{KeyGen}}

\providecommand{\Spark}{\mathsf{Spark}}
\providecommand{\Clone}{\mathsf{Clone}}
\providecommand{\Sign}{\mathsf{Sign}}
\providecommand{\Ver}{\mathsf{Ver}}
\providecommand{\cM}{\mathcal{M}}
\providecommand{\cO}{\mathcal{O}}
\providecommand{\AdvState}[1]{\ket{\mathrm{adv}_{#1}}}
\providecommand{\Tr}{\operationname{Tr}}
\providecommand{\E}{\mathbb{E}}
\providecommand{\bin}{\mathsf{bin}}

\title{
Separating $\cQMA$ and $\QCMA$ Relative to a Classical Oracle\\[0.5em]
\Large \emph{Clonability does not make quantum proofs and advice classical}
}

\author{
    Alper Cakan\thanks{Carnegie Mellon University, USA. \email{acakan@andrew.cmu.edu}. Work done while visiting Academia Sinica.}  \and
	Kai-Min Chung\thanks{Academia Sinica, Taiwan. \email{kmchung@iis.sinica.edu.tw}.}  \and
	Wei-Hsiang Hung\thanks{Academia Sinica, Taiwan. \email{arniehung@iis.sinica.edu.tw}.} \and
	Tzu-Yi Yang\thanks{Academia Sinica, Taiwan \& California Institute of Technology, USA. \email{tyang3@caltech.edu}.}
}

\date{\today}

\ifsubmit
\newcommand{\arnie}[1]{}
\newcommand{\red}[1]{}
\newcommand{\alper}[1]{}
\newcommand{\kaimin}[1]{}
\newcommand{\todo}[1]{}

\else
\newcommand{\arnie}[1]{$\ll$\textsf{\color{orange} Arnie: { #1}}$\gg$}
\newcommand{\red}[1]{\textcolor{red}{(Tzu-Yi: #1)}}
\newcommand{\alper}[1]{\textcolor{blue}{(Alper: #1)}}
\newcommand{\kaimin}[1]{\textcolor{teal}{(Kai-Min: #1)}}
\newcommand{\todo}[1]{\textcolor{blue}{(TODO: #1)}}

\fi

\begin{document}
\maketitle

\ifqip
  \input{111-main-qip}
\fi

\ifeprint
  % ePrint body; compile main.tex with \qipfalse.
\begin{abstract}
Since the introduction of the complexity class $\QMA$ as a quantum-verifier analogue of $\NP$ (Kitaev, 1997), many have wondered whether quantum proofs are necessary or classical proofs suffice -- that is, whether $\QMA = \QCMA$ or $\QCMA \subsetneq \QMA$ (Aharonov and Naveh, 2002; Aaronson and Kuperberg, CCC ’07). This longstanding question was recently answered by works of Bostanci, Haferkamp, Nirkhe, and Zhandry (STOC ’26) and Bostanci, Huang, and Vaikuntanathan (FOCS ’26), which showed that quantum proofs are more powerful than classical ones in the classical-oracle setting.

However, it remains unclear what exactly makes quantum proofs more powerful than classical ones. In the information-theoretic setting, a family of quantum states is \emph{not classicalizable} if and only if it is unclonable. Indeed, these recent works also explicitly highlight the unclonability of their quantum proofs as a key mechanism behind their separations, and their arguments crucially rely on this property. This raises the question of whether unclonability is necessary for quantum proofs to be more powerful than classical ones.

In this work, we show that even clonable quantum proofs can be more powerful than classical ones relative to a classical oracle by constructing a classical oracle $\cO$ such that $\QCMA^{\cO} \subsetneq \cQMA^{\cO}.$ This resolves the open question of Nehoran and Zhandry (ITCS ’24), who established the analogous separation relative to a quantum oracle. We also show a classical oracle separation between $\cBQPqpoly$ and $\BQPpoly$, and give applications of our results to quantum cryptography.
\end{abstract}

\newpage
\tableofcontents
\newpage

\section{Introduction}

% \begin{center}\begin{minipage}{0.8\linewidth}
%      \textit{Are clonable quantum proofs and advice strictly more useful than classical ones?}
% \end{minipage}\end{center}

% In this paper, we give an affirmative answer. At a high level, this result indicates that unclonability is \emph{not} the only essence that makes a quantum state different to a classical one.

% \paragraph{Hidden-Subspace Based vs. Recoverable-Code Based Construction.} \arnie{Why do we bring this up here? Are we suggesting that the QCMA hardness of our code-based construction has implications for the crypto applications listed here?}
% Previous constructions of quantum money \cite{AC:2012}, quantum fire \cite{CGS:2026} and one-shot signature \red{Add more reference...} are all based on \emph{hidden subspace}. The banknote state and the flame state are uniform superpositions of all vectors of a subspace or coset. Their security (copy-protection or non-transferring) is proven against $\BQP$-adversaries but not $\QCMA$-adversaries. This is because one can describe a subspace classically by their basis.

% \paragraph{Self-Verification Mechanism.}

Understanding the quantum advantage is a central
goal of quantum computation. In this paper, we focus on the following questions:

\begin{center}\begin{minipage}{0.8\linewidth}
      \textit{What makes quantum proofs more powerful than classical ones?\\Is unclonability necessary to the quantum advantage?}
\end{minipage}\end{center}

Classical proofs can be copied and reused freely, whereas arbitrary quantum states cannot. It is therefore natural to ask whether the additional power of a quantum proof fundamentally relies on this unclonability. A complementary motivation comes from the study of computational cloning and telegraphing. Nehoran and Zhandry \cite{NZ:2024} showed that, under computational constraints, efficient clonability need not imply the ability to convert a quantum state into reusable classical information, and introduced $\cQMA$ as the corresponding complexity-theoretic notion. Thus,
\[
    \QCMA \subseteq \cQMA \subseteq \QMA,
\]
placing clonable quantum proofs between classical and unrestricted quantum proofs.

The two endpoints of this hierarchy, $\QCMA$ and $\QMA$, have been the subject of a long-standing question. First posed by Aharonov and Naveh \cite{AN02}, the question of whether quantum proofs are more powerful than classical ones---that is, whether $\QCMA=\QMA$---has a long history. Following a sequence of relativized, restricted, and conditional results \cite{AK07, FK18, BFM23, NN24, LLPY24, BDK24, Zha25, LMY25}, recent breakthroughs established unconditional separations between $\QMA$ and $\QCMA$ relative to a classical oracle. Bostanci, Haferkamp, Nirkhe, and Zhandry \cite{BHNZ26} obtained the first such separation, and Bostanci, Huang, and Vaikuntanathan \cite{BHV:2026} subsequently gave a simpler proof based on the Yamakawa--Zhandry code-intersection problem and strong list-recoverable codes.

The recent classical oracle separation \cite{BHNZ26, BHV:2026} suggests that clonability may be tied to the the source of the QMA--QCMA gap. A common feature of their lower-bound arguments is that a classical witness can be copied and reused across many executions, eventually allowing one to extract too much information about the underlying oracle. This perspective naturally suggests that the same reasoning might extend to quantum witnesses that are themselves efficiently clonable. Indeed,~\cite{BHNZ26} and the first version
of~\cite{BHV:2026} suggest extensions to clonable variants
of their lower bounds; we discuss the additional
difficulty posed by oracle-assisted cloning in
Section~\ref{subsec:open-problems}.

A complementary line of evidence comes from Nehoran and Zhandry \cite{NZ:2024}, who directly studied the intermediate class $\cQMA$. They showed, relative to a unitary quantum oracle $U$, that
\[
    \cQMA^U\not\subseteq\QCMA^U.
\]
Thus, in a relativized setting, efficient clonability alone does not force quantum proofs to become classical. Their construction, however, relies on a quantum oracle and the clonabillity is direclty baked into the oracle, leaving open whether the same phenomenon can arise in a purely classical relativized world.

There is encouraging evidence from quantum cryptography that a quantum oracle may not be necessary for separating clonability from classicality. In the language of quantum fire, the analogous goal is to construct quantum states that can be efficiently cloned but cannot be efficiently telegraphed---that is, converted into classical information sufficient to reproduce the state. Çakan, Goyal, and Shmueli \cite{CGS:2026} obtain such a separation relative to a classical oracle. While this is a cryptographic rather than complexity-theoretic statement, it shows that a classical oracle can already support quantum information that is efficiently clonable yet resists classicalization.

This leads us to the following question:

\begin{center}
\begin{minipage}{0.8\linewidth}
\textit{Does there exist a classical oracle $\cO$ such that
$\cQMA^{\cO} \not\subseteq \QCMA^{\cO}$?}
\end{minipage}
\end{center}

We show that this phenomenon persists relative to a classical oracle. Moreover, our technique extends from quantum proofs to quantum advice: relative to a classical oracle, polynomial-size advice that can be efficiently cloned is strictly more useful than polynomial-size classical advice. Thus, in both settings, efficient clonability does not eliminate the advantage of quantum information.

The same construction also has cryptographic applications. Building on the Yamakawa--Zhandry algorithm~\cite{YZ:2022} and the good-code framework of~\cite{BHV:2026}, we obtain a quantum fire relative to a classical oracle, with a proof of efficient cloning and noninteractive untelegraphability.
%We further propose a quantum money scheme and an one-shot signature, whose security remains conjectural. Together, these constructions suggest a \emph{good-code based} approach to quantum cryptography.\alper{maybe remove}

\subsection{Our Results}

We answer the preceding question affirmatively: clonable quantum proofs can be strictly more powerful than classical proofs, even relative to a classical oracle.

\begin{restatable}{theorem}{mainseparationQMA}
\label{thm:main-separationQMA}
There exists a classical oracle $\cO$ such that
\[
    \cQMA^{\cO} \not\subseteq \QCMA^{\cO}.
\]
\end{restatable}

%At a high level, our result suggests that unclonability need not be essential to the advantage of quantum proofs over classical ones: relative to a classical oracle, even efficiently clonable quantum proofs can be more powerful than classical proofs.
This strengthens the quantum-oracle separation of~\cite{NZ:2024} to a classical oracle. Our honest witnesses admit polynomial-time cloning with fidelity $1-\operatorname{negl}(\lambda)$. Thus, unclonability is not essential to quantum-proof advantage in a relativized setting.

The same technique extends to clonable quantum advice

\begin{restatable}{theorem}{mainseparationBQP}
\label{thm:main-separationBQP}
There exists a classical oracle $\cO$ such that
\[
    \cBQPqpolyO \not\subseteq \BQPpolyO.
\]
\end{restatable}

Our constructions combine the code-intersection states of~\cite{YZ:2022,BHV:2026} with locked access to hash values, enabling quantum cloning while preserving hardness against classical proofs and advice.

We also obtain a good-code-based construction of quantum fire.
\begin{theorem}[Quantum Fire]
    Relative to a classical oracle, there is a quantum fire scheme satisfying strong clonability and untelegraphable security.
\end{theorem}
This gives a direct proof of the classical-oracle quantum-fire existence result of~\cite{CGS:2026}, without claiming their stronger key-fire or interactive-security guarantees.

%Tzu-Yi: I can write the below.

\subsection{Technical Overview}
\label{subsec:technical-overview}

Our construction builds on the code-intersection problem of
Yamakawa and Zhandry~\cite{YZ:2022}.
Given a linear code $C\subseteq\Sigma^n$, hash functions
$H_i:\Sigma\to\{0,1\}$, and a target
$b\in\{0,1\}^n$, the task is to find a codeword $c\in C$
such that
\[
H(c):=H_1(c_1)\|\cdots\|H_n(c_n)=b.
\]
For suitable codes, their quantum algorithm solves this
problem efficiently, whereas polynomial-query classical
algorithms fail. We use the biased variant
of~\cite{BHV:2026}, with targets
$b=z\|0^{n-\ell}$ and hash functions strongly biased
toward zero.

The central quantum state in our work is
\[
\advstate
=
\left(\bigotimes_{i=1}^{n}\ket{H_i^{-1}(0)}\right)
\otimes
\left(\bigotimes_{i=1}^{n}\ket{H_i^{-1}(1)}\right)
\otimes\ket{\phi_2},
\qquad
\ket{\phi_2}
=
\frac{1}{\sqrt{|C|}}\sum_{c\in C}\ket{c}.
\]
Here $\ket{H_i^{-1}(b)}$ is the uniform superposition
over the corresponding hash preimages.
Under different oracle interfaces, this same state serves as:
\begin{itemize}
    \item a clonable quantum \emph{proof} in $\cQMA$
          (\Cref{sec:QMA});
    \item a clonable quantum \emph{advice} in $\cBQPqpoly$
          (\Cref{sec:BQP});
    \item a \emph{flame state} in quantum fire
          (\Cref{sec:crypto-applications});
    %\item a \emph{banknote} in quantum money (\Cref{subsec:money-construction});
    %\item a \emph{signing key} in one-shot signature (\Cref{subsec:oss-construction}).
\end{itemize}
Because this state originates from an algorithm for the
Yamakawa-Zhandry search problem, we consistently call it the
`advice' state, including when it plays one of
the other roles above.

\paragraph{The Yamakawa-Zhandry algorithm.}
For a target $b$, select the appropriate preimage register
for each coordinate to obtain
\[
\ket{\phi_1(b)}
:=
\bigotimes_{i=1}^{n}\ket{H_i^{-1}(b_i)}.
\]
Let $\QFT$ denote the Fourier transform over the
additive group $\Sigma^n$, and write
\[
\QFT\ket{\phi_1(b)}
=\sum_{e\in\Sigma^n}a_b(e)\ket{e}.
\]
Fourier duality maps the uniform superposition over $C$
to the uniform superposition over $C^\perp$.
The algorithm then performs the following schematic
transformation:
\[
\begin{aligned}
\ket{\phi_1(b)}_A\ket{\phi_2}_B
&\xrightarrow{\ \QFT_A\otimes\QFT_B\ }
  \frac{1}{\sqrt{|C^\perp|}}
  \sum_{\substack{e\in\Sigma^n\\u\in C^\perp}}^{\phantom{\big|}}
  a_b(e)\ket{e}_A\ket{u}_B\\[-1mm]
&\xrightarrow{\ \mathrm{add}\ }
  \frac{1}{\sqrt{|C^\perp|}}
  \sum_{\substack{e\in\Sigma^n\\u\in C^\perp}}
  a_b(e)\ket{u+e}_A\ket{u}_B\\
&\xrightarrow{\ \mathrm{decode}\ }
  \left(
    \frac{1}{\sqrt{|C^\perp|}}
    \sum_{\substack{e\in\Sigma^n\\u\in C^\perp}}
    a_b(e)\ket{u+e}_A
  \right)\ket{0}_B\\
&\xrightarrow{\ \QFT_A^{-1}\ }
  \ket{\phi_1(b)\cdot\mathbf{1}_C}_A\ket{0}_B.
\end{aligned}
\]
In the first line, the summation means
$e\in\Sigma^n$ and $u\in C^\perp$.
The last two transformations hold approximately, and
the final state is understood to be normalized.
Addition implements convolution in the Fourier domain;
decoding $u+e$ recovers $u$ and allows its register to be
erased. The inverse Fourier transform then produces a
state supported on codewords hashing to $b$.
The bias of $H$, together with the short nonzero prefix
of $b$, makes the error $e$ sufficiently sparse for
dual decoding~\cite{BHV:2026}.

With unrestricted coherent access to $H$, a quantum
polynomial-time algorithm can prepare the required
preimage states itself. This is the setting underlying
the original quantum-versus-classical query
advantage~\cite{YZ:2022}.
To obtain separations involving quantum proofs or advice,
we instead restrict how algorithms can access $H$.

We consider two oracle interfaces:
\begin{enumerate}
    \item The verification oracle $O[H,E]$
          of~\cite{BHV:2026}, which only checks whether
          a supplied code-intersection pair is valid
          and belongs to $E$;
          %With $E$ full, this interface underlies our candidate quantum money and signing constructions.

    \item Our oracle $M[H,E]$, which additionally provides
          \emph{locked access} to hash values:
          \[
          M[H,E](i,x,c)
          =
          \begin{cases}
          H_i(x),
            &c\in C,\ 
             H(c)=\bin(i,x)\|0^{n-\ell},
             (\bin(i,x),c)\in E,\\
          \perp,&\text{otherwise}.
          \end{cases}
          \]
          $\bin(i,x)$ is the binary expression of $i\in [n]$ and $x\in\Sigma$. When $E$ is full, we denote $M[H]:=M[H,E]$, and the advice state can unlock
          this access and thereby reproduce itself.
          This interface is the basis of our
          $\cQMA$, $\cBQPqpoly$, and quantum-fire
          constructions,
\end{enumerate}
where
$$E\subseteq \{0,1\}^{\ell(\lambda)}\times \Sigma^n$$
is a subset, and $\ell=\lambda$ for $O$ and $\ell=\Theta(\lambda\log \lambda)$ for $E$. We define YES instances when $E$ is full and NO instances when the size of $E$ is smaller than some constant fraction.

For brevity, we focus below on the separation between
$\cQMA$ and $\QCMA$.

\paragraph{Clonability of $\advstate$.}
The main observation is that solving the code-intersection
problem coherently lets us extract hash values while
approximately preserving the advice.
Fix a coordinate $i$ and prepare a uniform superposition
over $x\in\Sigma$. Controlled on $x$, run a purification
$U^{\mathrm{YZ}}$ of the biased YZ algorithm on the
challenge $\bin(i,x)$, query $M[H]$ using
the resulting codeword, and then apply
$(U^{\mathrm{YZ}})^{-1}$.

On a successful YZ output, the oracle answer is always
$H_i(x)$, independently of which valid codeword was
produced. Thus the answer can be copied to a separate
register while the YZ computation is uncomputed.
Since YZ succeeds with overwhelming probability for
every challenge, this approximately implements
\[
\advstate\otimes
\frac{1}{\sqrt{|\Sigma|}}\sum_{x\in\Sigma}\ket{x}\ket{0}
\longmapsto
\advstate\otimes
\underbrace{
\frac{1}{\sqrt{|\Sigma|}}
\sum_{x\in\Sigma}\ket{x}\ket{H_i(x)}
}_{\displaystyle \ket{\psi_i^H}}.
\]
The approximation includes preservation of the advice
and its decoupling from the extracted registers.

We repeat this procedure polynomially many times using
the same seed advice. A hybrid argument bounds the
accumulated disturbance by a negligible quantity.
Measuring the last qubit of $\ket{\psi_i^H}$ produces
$\ket{H_i^{-1}(0)}$ or $\ket{H_i^{-1}(1)}$.
The smaller preimage set still occupies an inverse-polynomial
fraction of the alphabet, so polynomially many repetitions
suffice to obtain both states for every coordinate with
overwhelming probability.
Finally, $\ket{\phi_2}$ can be prepared without oracle
queries because $C$ is an explicit linear code.
Assembling these registers gives a fresh advice state
while preserving the original one, with fidelity
$1-\operatorname{negl}(\lambda)$. An overview of the cloning algorithm is \Cref{fig:cloning_schematic}.

\begin{figure}[t]
\centering
\begin{tikzpicture}[
    >=Latex,
    font=\small,
    every node/.style={inner sep=2pt}
]

% ------------------------------------------------------------
% 3D operator box: name, x, y, width, height, depth, label
% ------------------------------------------------------------
\newcommand{\opBox}[7]{%
  \coordinate (#1A) at (#2,#3);
  \coordinate (#1B) at ($(#1A)+(#4,0)$);
  \coordinate (#1C) at ($(#1A)+(0,#5)$);
  \coordinate (#1D) at ($(#1B)+(0,#5)$);
  \coordinate (#1E) at ($(#1A)+(#6,#6)$);
  \coordinate (#1F) at ($(#1B)+(#6,#6)$);
  \coordinate (#1G) at ($(#1C)+(#6,#6)$);
  \coordinate (#1H) at ($(#1D)+(#6,#6)$);
  \coordinate (#1M) at ($(#1A)+({0.5*#4},{0.5*#5})$);

  \draw[thick,fill=white] (#1A)--(#1B)--(#1D)--(#1C)--cycle;
  \draw[thick,fill=gray!10] (#1C)--(#1D)--(#1H)--(#1G)--cycle;
  \draw[thick,fill=gray!5] (#1B)--(#1D)--(#1H)--(#1F)--cycle;
  \draw[thick] (#1A)--(#1E)--(#1F)--(#1B);
  \draw[thick] (#1E)--(#1G);
  \draw[thick] (#1G)--(#1H);

  \node at (#1M) {\footnotesize #7};

  \coordinate (#1in)    at ($(#1A)+(0,{0.5*#5})$);
  \coordinate (#1out)   at ($(#1B)+(0,{0.5*#5})$);
  \coordinate (#1south) at ($(#1A)+({0.5*#4},0)$);
}

% ------------------------------------------------------------
% Background regions
% ------------------------------------------------------------
\fill[blue!5,rounded corners=4pt]   (-0.70,-2.70) rectangle (1.00,1.95);
\fill[orange!6,rounded corners=4pt] (1.20,-2.70) rectangle (8.95,1.95);
\fill[green!6,rounded corners=4pt]  (9.15,-2.70) rectangle (14.75,1.95);

\node at (0.15,2.18) {\scriptsize Seed witness};
\node at (5.05,2.18) {\scriptsize Locked hash extraction};
\node at (11.75,2.18) {\scriptsize Reconstruction};

% ------------------------------------------------------------
% Advice flow
% ------------------------------------------------------------
\node[draw,rounded corners=3pt,fill=white,minimum width=0.88cm] (advIn) at (0.15,0.95)
  {$|\mathrm{adv}_{H}\rangle$};

\opBox{Uone}{1.55}{0.54}{0.92}{0.82}{0.16}{$U_1$}
\opBox{Utwo}{3.60}{0.54}{0.92}{0.82}{0.16}{$U_2$}
\opBox{UR}{7.58}{0.54}{0.92}{0.82}{0.16}{$U_R$}

% Top dots centered exactly between U_2 and U_R
\coordinate (dotL) at (5.65,0.95);
\coordinate (dotR) at (6.45,0.95);
\node at (6.05,0.95) {$\cdots$};

\node[draw,rounded corners=3pt,fill=white,minimum width=0.88cm] (advKeep) at (12.05,0.95)
  {$|\mathrm{adv}_{H}\rangle$};
\node[draw,rounded corners=3pt,fill=white,minimum width=0.88cm] (advNew) at (12.05,-1.23)
  {$|\mathrm{adv}_{H}\rangle$};

\draw[->,thick] (advIn.east) -- (Uonein);
\draw[->,thick] (Uoneout) -- (Utwoin);
\draw[->,thick] (Utwoout) -- (dotL);
\draw[->,thick] (dotR) -- (URin);
\draw[->,thick] (URout) -- node[above,font=\scriptsize] {preserved} (advKeep.west);

% ------------------------------------------------------------
% Extracted hash states
% ------------------------------------------------------------
\node[draw,rounded corners=3pt,fill=white,minimum width=0.86cm] (psi1) at (2.01,-0.78)
  {$|\psi_{i_1}\rangle$};
\node[draw,rounded corners=3pt,fill=white,minimum width=0.86cm] (psi2) at (4.06,-0.78)
  {$|\psi_{i_2}\rangle$};

\node at (6.05,-0.78) {$\cdots$};

\node[draw,rounded corners=3pt,fill=white,minimum width=0.86cm] (psiR) at (8.04,-0.78)
  {$|\psi_{i_R}\rangle$};

\draw[->,thick] (Uonesouth) -- (psi1.north);
\draw[->,thick] (Utwosouth) -- (psi2.north);
\draw[->,thick] (URsouth) -- (psiR.north);

% ------------------------------------------------------------
% Reconstruction box M
% ------------------------------------------------------------
% Order of the four equally spaced levels near M:
%   1. psi_R   (visible)
%   2. dots    (invisible omitted wire)
%   3. psi_2   (visible)
%   4. psi_1   (visible)
%
% This swaps the dots and the psi_2 wire compared to before.

\opBox{Mbox}{9.55}{-2.04}{1.35}{1.62}{0.17}{$\mathsf M$}

% Four equally spaced input levels
\coordinate (MinTop)    at ($(Mboxin)+(0,0.45)$);   % level 1: psi_R
\coordinate (MinDots)   at ($(Mboxin)+(0,0.15)$);   % level 2: dots
\coordinate (MinPsiTwo) at ($(Mboxin)+(0,-0.15)$);  % level 3: psi_2
\coordinate (MinBot)    at ($(Mboxin)+(0,-0.45)$);  % level 4: psi_1

% Common x-position before entering M
\coordinate (preM) at ($(Mboxin)+(-0.42,0)$);

% ------------------------------------------------------------
% Top wire: psi_R -> M (horizontal)
% ------------------------------------------------------------
\draw[->,thick] (psiR.east) -- (MinTop);

% ------------------------------------------------------------
% Middle wire: psi_2 -> M
% Swapped below the dots so its vertical part is longer
% ------------------------------------------------------------
\coordinate (pTwoA) at ($(psi2.east)+(0.62,0)$);
\coordinate (pTwoB) at ($(pTwoA |- MinPsiTwo)$);
\coordinate (pTwoC) at ($(preM  |- MinPsiTwo)$);

\draw[->,thick,rounded corners=8pt,line cap=round,line join=round]
  (psi2.east) -- (pTwoA) -- (pTwoB) -- (pTwoC) -- (MinPsiTwo);

% ------------------------------------------------------------
% Vertical dots: exactly between the top and middle wires, near M
% ------------------------------------------------------------
\coordinate (dotPos) at ($(preM |- MinDots)+(0.10,0)$);
\foreach \dy in {-0.09,0,0.09}{
  \fill ($(dotPos)+(0,\dy)$) circle (0.55pt);
}

% ------------------------------------------------------------
% Bottom wire: psi_1 -> M
% ------------------------------------------------------------
\coordinate (pOneA) at ($(psi1.east)+(0.58,0)$);
\coordinate (pOneB) at ($(pOneA |- MinBot)$);
\coordinate (pOneC) at ($(preM  |- MinBot)$);

\draw[->,thick,rounded corners=8pt,line cap=round,line join=round]
  (psi1.east) -- (pOneA) -- (pOneB) -- (pOneC) -- (MinBot);

% ------------------------------------------------------------
% Output arrow from M to fresh advice
% ------------------------------------------------------------
\draw[->,thick] (Mboxout) -- (advNew.west);

% ------------------------------------------------------------
% Output brace and label
% ------------------------------------------------------------
\draw[decorate,decoration={brace,mirror,amplitude=5pt},thick]
  (12.85,-1.55) -- (12.85,1.22)
  node[midway,right=7pt] {$|\mathrm{adv}_{H}\rangle^{\otimes 2}$};

% ------------------------------------------------------------
% Short formula below extraction part
% ------------------------------------------------------------
\node[align=center,font=\scriptsize] at (5.10,-2.30)
  {$\displaystyle
  |\psi_i\rangle :=
  \frac{1}{\sqrt{|\Sigma|}}
  \sum_{x\in\Sigma}|x\rangle |H_i(x)\rangle
  $};

\end{tikzpicture}

\caption{\textit{Schematic of the cloning algorithm.}
Starting from one copy of \(|\mathrm{adv}_H\rangle\), each locked hash-extraction round
\(U_j\) approximately preserves \(|\mathrm{adv}_H\rangle\) while producing
\(|\psi_{i_j}\rangle = |\Sigma|^{-1/2}\sum_{x\in\Sigma}|x\rangle|H_{i_j}(x)\rangle\).
After \(R=\mathrm{poly}(\lambda)\) repetitions, the reconstruction map \(\mathsf M\)
postselects the states \(|H_i^{-1}(0)\rangle\) and \(|H_i^{-1}(1)\rangle\), and assembles
a fresh copy of \(|\mathrm{adv}_H\rangle\).}
\label{fig:cloning_schematic}
\end{figure}
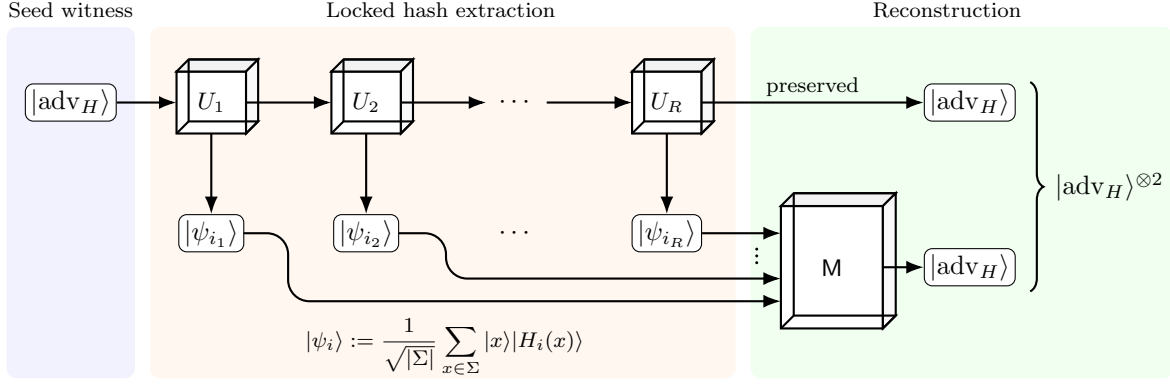

\paragraph{QCMA-hardness relative to $M[H,E]$.}
We follow the query-extraction approach underlying the
separations of~\cite{BHNZ26,BHV:2026}.
Suppose a polynomial-query verifier with a classical
witness distinguishes the full oracle $M[H]$ from all
allowed restricted oracles $M[H,E]$.
The hybrid lemma forces its queries to place noticeable
weight on valid pairs omitted from $E$.
Measuring a uniformly selected query therefore extracts
a new valid pair with inverse-polynomial probability.

The crucial step is to make this extraction
\emph{repeatable without oracle access}.
After extracting a pair, we insert it into the simulated
set $E$ and restart the verifier with the same classical
witness. The next hybrid comparison then forces query
weight onto valid pairs outside the updated set.
To maintain this argument, however, we must simulate
the complete oracle answers on all previously inserted
pairs. Unlike BHV's verification oracle, our oracle
also returns a hash bit.

This is why we encode both $i$ and $x$ into the challenge
$z=\bin(i,x)$.
Each stored valid pair then unlocks only one hash value,
$H_i(x)$, which the no-query simulator guesses and records.
A fresh guess costs only a factor of $1/2$ per round.
If the challenge encoded only $x$, the same stored
codeword could unlock $H_j(x)$ for every $j\in[n]$.
Simulating subsequent executions could then require
guessing an entire $n$-bit hash column, introducing a
per-round loss too large for this argument.
Including the coordinate index keeps the simulation
cost compatible with the lower bound.

Consequently, guessing the classical witness and
performing $L$ extraction rounds yields a no-query
guesser whose success probability is at least
\[
2^{-t(\lambda)}
\left(\frac{1}{\operatorname{poly}(\lambda)}\right)^L.
\]
On the other hand, the near-optimal list recovery of
our code forces $L$ distinct valid pairs to prescribe
hash values at $\Omega(nL)$ distinct coordinate-symbol
locations. Under the biased hash distribution, any
no-query guesser succeeds with probability at most
\[
(1-\lambda^{-4})^{\Omega(nL)}
=\exp(-\Omega(\lambda L)).
\]
Choosing $L$ large enough to absorb the witness-guessing
cost makes these bounds incompatible.
Thus the locked hash access enables efficient quantum
cloning while preserving the classical-witness
lower bound.

\subsection{Open Problems}
\label{subsec:open-problems}

\paragraph{Separating $\QMA$ from $\cQMA$.}
Chatterjee, Kundu, and Podder~\cite{CKP:2025} establish
a quantum-oracle separation between $\QMA$ and $\cQMA$,
through the stronger notion of strictly uncloneable
proofs. Our conjecture below would strengthen the
$\QMA$-versus-$\cQMA$ separation from a quantum oracle
to a classical oracle.

Bostanci, Haferkamp, Nirkhe, and Zhandry~\cite{BHNZ26}
suggest that their sampling argument extends to this
setting, and the first version of~\cite{BHV:2026}
makes a similar remark.
However, extending the classical-witness argument to
oracle-assisted quantum cloning requires an additional
step. A classical witness can be reused without oracle
queries, whereas a quantum cloning algorithm may
query the hidden oracle. These queries remain even
when the verifier's queries are replaced by a simulated
restricted oracle. Furthermore, cloning correctness
under the original oracle need not persist under
this replacement. Thus, clonability alone does not
yield the no-query guesser needed by that argument.
This obstacle disappears for query-free cloners,
but $\cQMA$ permits oracle-assisted cloning.

A possible approach is to extend the guessing bound
of~\cite[Lemma~6.6]{BHV:2026} to algorithms making
polynomially many queries to $O[H,E]$, accounting for
both cloning and sample extraction.
We have not obtained a sufficiently strong bound.
We conjecture that BHV's code intersection subset
size problem lies outside $\cQMA$, which would imply
the following separation.

\begin{conjecture}
There exists a classical oracle $O$ such that
\[
\QMA^O\not\subseteq\cQMA^O.
\]
\end{conjecture}

\begin{figure}[t]
\centering
\begin{tikzpicture}[
    scale=0.8,
    transform shape,
    class boundary/.style={draw=black, thick},
    annotation/.style={font=\small, align=left},
    pointer/.style={
        -{Stealth[length=2mm]},
        semithick,
        shorten <=2pt
    }
]
    % Each inner ellipse is shifted downward.
    \draw[class boundary]
        (0,0) ellipse (4 and 2.6);
    \draw[class boundary]
        (0,-0.35) ellipse (2.65 and 1.8);
    \draw[class boundary]
        (0,-0.7) ellipse (1.35 and 0.95);

    % Labels placed in the enlarged upper gaps.
    \node at (0,2.02) {$\QMA$};
    \node at (0,0.7) {$\cQMA$};
    \node at (0,-0.7) {$\QCMA$};

    % Our problem: inside clonableQMA, outside QCMA.
    \coordinate (ours) at (1.9,-0.5);
    \fill (ours) circle (2pt);
    \draw[pointer]
        (ours) -- (5.15,0.55);
    \node[annotation, anchor=west]
        at (5.25,0.55) {Our problem};

    % Conjectured placement: inside QMA, outside clonableQMA.
    \coordinate (bhv) at (3.15,-0.95);
    \fill (bhv) circle (2pt);
    \draw[pointer, dashed]
        (bhv) -- (5.15,-1.45);
    \node[annotation, anchor=west]
        at (5.25,-1.45) {BHV's problem\\(conjectured)};
\end{tikzpicture}
\caption{
The hierarchy
$\QCMA\subseteq\cQMA\subseteq\QMA$.
The clonable code-intersection subset size problem
separates $\cQMA$ from $\QCMA$ relative to a classical
oracle.
We conjecture that BHV's code intersection subset
size problem~\cite{BHV:2026} separates $\QMA$ from
$\cQMA$ relative to a classical oracle.
A quantum-oracle
separation is already known~\cite{CKP:2025}.
}
\label{fig:clonable-qma-hierarchy}
\end{figure}
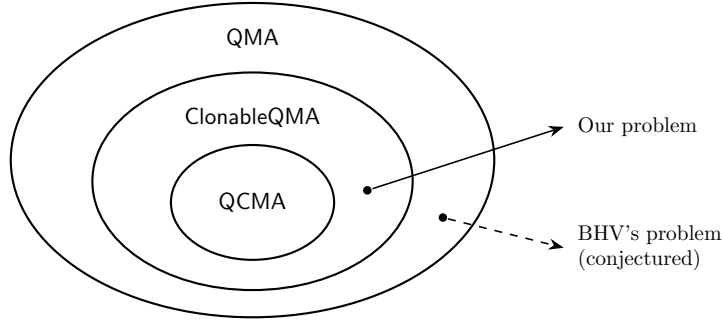

%\paragraph{Good-code-based quantum cryptography.}\alper{maybe remove?}
%Our construction provides a good-code-based approach to quantum cryptography: the same code-intersection advice state supports different functionalities depending on the classical oracle interface. We prove quantum-fire security using the query-extraction and guessing bounds underlying our complexity separation. We also propose a quantum-money mini-scheme and an honestly generated signature token, whose security remains conjectural.

%For quantum money, the main challenge is to rule out the production of two accepted states from one banknote. This may be related to the conjectured separation between $\QMA$ and $\cQMA$, although we do not know a reduction. For the signature token, the corresponding challenge is to rule out signatures on two distinct messages from one issued state. Obtaining standard one-shot signatures additionally requires a generation mechanism secure against malicious key generation. Establishing these properties would extend the good-code-based approach beyond quantum fire.

\paragraph{Quantum fire beyond classical-witness hardness.}
Goldin, Morimae, Mutreja, and Yamakawa \cite{GMMY:2024}
study \emph{CountCrypt}, a class of quantum cryptographic
primitives that can survive $\BQP=\QCMA$ but are broken
when $\BQP=\mathsf{PP}$.
They construct a unitary-oracle world in which
$\BQP=\QCMA$ while quantum-computation-classical-communication
commitments and multiparty noninteractive key exchange
exist.
This motivates asking whether comparable phenomena
can be obtained relative to purely classical oracles,
and whether quantum fire can survive such a collapse.

For constructions whose flame states are uniform
superpositions over hidden linear subspaces, a basis
provides a polynomial-size classical description from
which the state can be prepared.
Whenever a proposed basis can be checked efficiently
using the public verification procedure, recovering
such a description becomes a classical-witness search
problem. This presents an obstacle to security against
$\QCMA$-powered reconstruction.

Our current construction has a related obstacle:
the hidden index $u$ is itself a succinct classical
preparation key.
Given the serial number $s=P(u)$, the predicate
\[
\exists\,v\in\{0,1\}^{\lambda}:
\quad v\text{ extends }a
\quad\land\quad P(v)=s
\]
belongs to $\operatorname{NP}$ relative to the public
oracle. If $\QCMA=\BQP$ relative to that oracle,
an adversary can decide these predicates and recover
$u$ one bit at a time.
It can then use $\mathbf{H}(u,\cdot,\cdot)$ to prepare
a fresh flame.
Consequently, the present construction does not
survive this collapse.

A natural open problem is therefore to design a
good-code-based generation mechanism that avoids
such classical preparation keys, and to prove its
security in a classical-oracle world satisfying
$\QCMA=\BQP$.
A $\mathsf{PSPACE}$-style oracle augmentation is one possible
route to a complexity collapse, but its effect on
security must be analyzed explicitly.
Such a result would connect quantum fire to the
classical-oracle questions motivated by CountCrypt;
it would not by itself establish that quantum fire
is a CountCrypt primitive.

\subsection{Organization}
In \Cref{sec:preliminaries}, we recall the quantum query model and
the complexity classes that we use.

Part~I contains our technical tools. In \Cref{tools:sec-coding}, we
define good code families and recall a family of multiplicity codes
satisfying the definition. In \Cref{tools:sec-code-intersection}, we
introduce the biased code intersection problem and prove a guessing
bound with classical side information. Then, we describe the code
intersection state and define good hashes for which it solves every
target under consideration. We also show that biased hashes are
good with overwhelming probability for our parameters. In
\Cref{tools:sec-cloning}, we define self-encrypted oracles and prove
that they allow the code intersection state to be cloned.

In Part~II, we prove our oracle separations. In \Cref{sec:QMA}, we
give a $\cQMA$ proof system for the self-encrypted CISS problem and
prove that no $\QCMA$ proof system solves it. Then, in
\Cref{sec:BQP}, we separate $\cBQPqpoly$ from $\BQPpoly$ by first
proving search hardness and then passing to decision using Yao's
box problem.

Finally, in \Cref{sec:crypto-applications}, we define a classical
oracle and give our quantum fire construction relative to it.
We state its verification correctness, cloning correctness, and
noninteractive untelegraphability guarantees there. The formal
definitions appear in \Cref{app:crypto-definitions}, and we give
the proof in \Cref{app:missing-security}.
\newpage\section{Preliminaries}
\label{sec:preliminaries}
\subsection{Notation}
We write $\operatorname{Pr}$ and $\mathbb{E}$ for probability
and expectation. A nonnegative function $\epsilon(\lambda)$
is negligible if, for every constant $c>0$,
$\epsilon(\lambda)\leq\lambda^{-c}$ for all sufficiently
large $\lambda$; we denote such functions by
$\operatorname{negl}(\lambda)$.

For density operators $\rho,\sigma$, write
\[
D(\rho,\sigma):=\frac12\|\rho-\sigma\|_1.
\]
Trace distance is contractive under quantum channels and
bounds the difference in the probability of any measurement
event. For a pure target state, we use the fidelity convention
\[
F(\rho,|\psi\rangle):=
\sqrt{\langle\psi|\rho|\psi\rangle}.
\]

\subsection{Quantum Query Model}
\label{subsec:quantum-query-model}

A classical oracle $f:\mathcal{X}\to\{0,1\}^{r}$ is accessed
through the unitary
\[
U_f|x,y\rangle=|x,y\oplus f(x)\rangle.
\]
A $T$-query quantum algorithm interleaves $T$ applications
of $U_f$ with oracle-independent quantum operations.
Measurements may be deferred by purification.
In query lower bounds, computation between queries is
unrestricted; in efficient algorithms, the entire computation
has polynomial size. All algorithms compared in a hybrid
argument start from the same state.

\begin{lemma}[Hybrid Lemma {\cite{BBBV:1997}}]
\label{lem:query-hybrid}
Let $f,g:\mathcal{X}\to\{0,1\}^{r}$ agree outside
$S\subseteq\mathcal{X}$.
Consider a $T$-query algorithm with final states
$\rho^f$ and $\rho^g$. Let $\rho_t^f$ be its state
immediately before its $t$th query when using $f$, and put
\[
p_t(S):=
\operatorname{Tr}\bigl[
  (\Pi_S\otimes I)\rho_t^f
\bigr],
\qquad
\Pi_S:=\sum_{x\in S}|x\rangle\langle x|.
\]
Then
\[
D(\rho^f,\rho^g)
\leq 2\sum_{t=1}^{T}\sqrt{p_t(S)}
\leq 2\left( T\sum_{t=1}^{T}p_t(S)\right)^{1/2}.
\]
In particular, the same bound applies to the difference
between the two acceptance probabilities.
\end{lemma}

\subsection{Hybrid Bound for Quantum Channels}
We need another lemma to compare the algorithmic outcome running by ideal channels and real channels.

\begin{lemma}[Hybrid Bound for Quantum Channels] \label{lem:hybrid-channel} Let $\Phi_1,\ldots,\Phi_T$ and $\widetilde{\Phi}_1,\ldots,\widetilde{\Phi}_T$ be quantum channels. Let $\rho_0=\widetilde{\rho}_0$ be a fixed initial state, define \[ \rho_t=\Phi_t(\rho_{t-1}), \qquad \widetilde{\rho}_t= \widetilde{\Phi}_t(\widetilde{\rho}_{t-1}), \] and suppose that \[ D\!\left( \Phi_t(\widetilde{\rho}_{t-1}), \widetilde{\Phi}_t(\widetilde{\rho}_{t-1}) \right) \le \epsilon_t. \] Then \[ D(\rho_T,\widetilde{\rho}_T) \le \sum_{t=1}^T\epsilon_t. \]
The same statement continues to hold if arbitrary quantum channels, including measurements, are interleaved between the subroutines.
\end{lemma}
\begin{proof} The triangle inequality and contractivity of trace distance under quantum channels give \begin{align*} D(\rho_t,\widetilde{\rho}_t) &= D\!\left( \Phi_t(\rho_{t-1}), \widetilde{\Phi}_t(\widetilde{\rho}_{t-1}) \right)\\ &\le D\!\left( \Phi_t(\rho_{t-1}), \Phi_t(\widetilde{\rho}_{t-1}) \right) + D\!\left( \Phi_t(\widetilde{\rho}_{t-1}), \widetilde{\Phi}_t(\widetilde{\rho}_{t-1}) \right)\\ &\le D(\rho_{t-1},\widetilde{\rho}_{t-1}) +\epsilon_t. \end{align*} Iterating this inequality proves \[ D(\rho_T,\widetilde{\rho}_T) \le \sum_{t=1}^T\epsilon_t. \]
Measurements and general intermediate quantum operations cause no difficulty, since they are quantum channels and trace distance is contractive under every quantum channel. \end{proof}

\subsection{Complexity Theory}
\label{subsec:complexity-theory}

All quantum algorithms below are uniform polynomial-time
algorithms. Witnesses and advice have polynomially many
bits or qubits.

\begin{definition}[$\QMA$, $\QCMA$ and $\cQMA$]
A promise problem $\mathcal{L}=(\mathcal{L}_{\mathrm{Yes}},\mathcal{L}_{\mathrm{No}})$ belongs to $\QMA$ if there are a
verifier $V$ and a polynomial $m$ such that, for
$\lambda=|x|$,
\[
\begin{aligned}
x\in \mathcal{L}_{\mathrm{Yes}}
&\implies
\exists\text{ $m(\lambda)$-qubit proof }\ket{\psi_x}\text{ such that }
\operatorname{Pr}[V(x,|\psi_x\rangle)=1]\geq \frac23,\\
x\in \mathcal{L}_{\mathrm{No}}
&\implies
\forall\text{ $m(\lambda)$-qubit }\ket{\psi_x}, \operatorname{Pr}[V(x,\rho)=1]\leq \frac13.
\end{aligned}
\]
The class $\QCMA$ is defined analogously, with the proof
restricted to a classical string
$w_x\in\{0,1\}^{m(\lambda)}$.

The class $\cQMA$ additionally requires a polynomial-time
cloning algorithm $\operatorname{Clone}$ such that, for every
$x\in L_{\mathrm{YES}}$, there is a proof
$|\psi_x\rangle$ satisfying both completeness and
\[
F\left(
  \operatorname{Clone}
    (x,|\psi_x\rangle\langle\psi_x|),
  |\psi_x\rangle^{\otimes 2}
\right)
\geq 1-\negl(\lambda).
\]
Here the cloning algorithm outputs two witness registers,
with all other registers traced out.
\end{definition}

Cloning is required only for the designated honest
witnesses; soundness still quantifies over every witness.
This is the negligible-error fidelity version of
clonable quantum proofs considered in~\cite{NZ:2024}. Obviously,
\[
\QCMA\subseteq\cQMA\subseteq\QMA.
\]

\begin{definition}[$\BQPpoly$, $\BQPqpoly$ and $\cBQPqpoly$]
A promise problem $\mathcal{L}=(\mathcal{L}_{\mathrm{Yes}},\mathcal{L}_{\mathrm{No}})$ belongs to $\BQPqpoly$ if there
are an algorithm $A$ and a family of polynomial-size
advice states $\{|\alpha_\lambda\rangle\}_{\lambda\in\mathbb{N}}$
such that, for every promised input $x$ of length $\lambda$,
\[
\operatorname{Pr}
 \bigl[A(x,|\alpha_\lambda\rangle)=L(x)\bigr]
\geq \frac23,
\]
where $L(x)=1$ on $\mathcal{L}_{\mathrm{Yes}}$ and $L(x)=0$ on
$\mathcal{L}_{\mathrm{No}}$.
The advice depends only on the input length, and need
not be efficiently preparable.

The class $\BQPpoly$ restricts this advice to a
polynomial-length classical string.
The class $\cBQPqpoly$ additionally requires a
polynomial-time algorithm $\operatorname{Clone}$ satisfying
\[
F\left(
 \operatorname{Clone}
   (1^\lambda,|\alpha_\lambda\rangle
                    \langle\alpha_\lambda|),
 |\alpha_\lambda\rangle^{\otimes 2}
\right)
\geq 1-\negl(\lambda).
\]
as above.
Removing advice altogether gives $\BQP$.
\end{definition}

Hence,
\[
\BQP\subseteq\BQPpoly
\subseteq\cBQPqpoly
\subseteq\BQPqpoly.
\]

\paragraph{Classical oracle relativization.}
For a fixed classical oracle $O:\{0,1\}^*\to\{0,1\}$,
the relativized classes are defined by allowing every
verifier, decision algorithm, and cloning algorithm
coherent access to $U_O$.
Witnesses may depend on both $x$ and $O$, while advice
may depend on $O$ and the input length, but not on the
particular input of that length.
The underlying oracle remains a classical function,
although queries to it may be quantum.

\paragraph{Diagonalization and clonability.}
Our lower bounds are first formulated as oracle-input
problems: the instance is specified by the oracle supplied
to the algorithm. Standard diagonalization places suitable
instances in disjoint blocks of one fixed oracle and
encodes their labels as classical inputs, yielding the
claimed relativized language or promise-problem
separations. We omit these details, following the
diagonalization framework of~\cite[Appendix~B]{BHV:2026}.

Clonability is preserved by this construction. Indeed, the global
cloning algorithm uses the same honest state and simulates each
query to the instance oracle by accessing its block. This coherent
simulation is exact, so it leaves the cloning channel and its
fidelity bound unchanged. We use efficiently computable block
addresses and input lengths polynomially related to the original
parameter. Thus, the running time and state size remain polynomial,
and the error remains negligible.

In the advice setting, we place at most one instance at each
selected input length, so the same advice state serves every input
of that length. The cloner can therefore locate the block using
only the input length. At unused lengths, we use a fixed basis
state, which can be copied exactly.\footnote{A classical flag in
the advice identifies unused lengths and is copied along with the
state; the decision algorithm rejects at those lengths. The selected
advice blocks can be chosen to have good hashes: otherwise, choosing
successful classical advice for every good block would give average
success at least $\frac23(1-2^{-2\lambda})>\frac35$, contradicting
\Cref{thm:advice-oracle-input}. In the proof setting, good hashes are
already part of the YES promise, and the input identifies the block.}

%\section{Separating $\QMA$ from $\cQMA$}
%\subsection{Biased Compressed Oracle}
%\subsection{Reduction from $\cQMA$-Hardness to Biased Inversion Problem}

\newpage\part{Technical Tools}
\section{Coding Theory}
\label{tools:sec-coding}

In this section, we collect the coding-theoretic properties that
we will use. First, we define a good code family by its construction,
decoding, and combinatorial properties. Then, we recall a family of
multiplicity codes satisfying these properties.

Let $C\subseteq\Sigma^n$, where $\Sigma=\mathbb F_q^s$, be an
$\mathbb F_q$-linear code. Define its dual using the standard
$\mathbb F_q$-bilinear pairing:
\[
C^\perp
 :=\left\{
 u\in\Sigma^n:
 \sum_{i=1}^{n}\sum_{j=0}^{s-1}u_{i,j}c_{i,j}=0
 \text{ for every }c\in C
 \right\}.
\]
Throughout, Hamming weight and distance count nonzero
$\Sigma$-symbols, rather than individual field elements.

We now give the definition of a good code family. We leave its
parameters general here and specify them in the existence theorem
below.

\begin{definition}[Good code family]\label{tools:def-good-code}
Let $n,q,s,r,L_{\max}$ be functions of a parameter $\lambda$, with
$n,s\ge1$ integers, $q$ a prime power,
$0\le r\le n$ and $L_{\max}\ge1$ integers,
and $n,s,\log q=\operatorname{poly}(\lambda)$.
A family of linear codes
$C_\lambda\subseteq\Sigma^n=(\mathbb F_q^s)^n$ is
\emph{$(r,L_{\max})$-good} if, for all sufficiently large $\lambda$,
it has the following properties.
\begin{enumerate}
    \item \textbf{Efficient construction and dual decoding.}
    A uniform polynomial-time algorithm constructs the field
    representation and a full-rank generator matrix of $C_\lambda$.
    Field arithmetic is efficient, and there is a uniform
    polynomial-time decoder satisfying
    \[
    \operatorname{Dec}_{C_\lambda^\perp}(u+e)=u
    \quad\text{for every }u\in C_\lambda^\perp
    \text{ and every }e\in\Sigma^n\text{ with }\operatorname{wt}(e)\le r.
    \]
    \item \textbf{List recovery.} For every integer
    $1\le L\le L_{\max}$, $C_\lambda$ is $(L,2L)$-list recoverable:
    if $S_1,\ldots,S_n\subseteq\Sigma$ satisfy
    \[
    \frac1n\sum_{i=1}^n|S_i|\le L,
    \]
    then at most $2L$ codewords $c\in C_\lambda$ satisfy
    $c_i\in S_i$ for every $i\in[n]$.
    \item \textbf{Coordinate-symbol diversity.} For every integer
    $1\le L\le L_{\max}$ and every $L$ distinct codewords
    $c^{(1)},\ldots,c^{(L)}\in C_\lambda$,
    \begin{equation}
    \sum_{i=1}^n\bigl|\{c_i^{(j)}:j\in[L]\}\bigr|\ge\frac{nL}{2}.
    \label{tools:eq-code-diversity}
    \end{equation}
\end{enumerate}
The last condition is an explicit expansion requirement; we do not
infer its exact constant from list recovery at integer list sizes.
\end{definition}

\begin{remark}
The generator matrix also permits preparation of
$|C_\lambda|^{-1/2}\sum_{c\in C_\lambda}|c\rangle$ without oracle
queries, to trace-distance error $\delta$ in time
$\operatorname{poly}(\lambda,\log(1/\delta))$.
Prepare uniform field coefficients, apply the linear encoding
reversibly, and uncompute the coefficient register using a left
inverse of the encoding obtained by Gaussian elimination.
\end{remark}

We now recall the code family of~\cite{BHV:2026} and verify that it
satisfies the definition above.

\begin{theorem}[Existence of good code families]
\label{tools:thm-good-code-existence}
We use the code and code parameters of~\cite{BHV:2026}.
For sufficiently large $\lambda$, choose $q$ to be the smallest
prime greater than $\lambda^5$, and set
\[
\begin{array}{lll}
    \lambda^5\le q\le2\lambda^5 &\text{a prime number,}&\\
    k=\lambda^3 &\text{the message dimension,}&\\
    s=\lambda &\text{the number of field elements in each symbol,}&\\
    n=q &\text{the codeword length,}&\\
    \Sigma=\mathbb F_q^s &\text{the alphabet.}&
\end{array}
\]
For a polynomial $f\in\mathbb F_q[X]$, its $j$th Hasse derivative
$f^{(j)}$ is defined by
\[
f(X+T)=\sum_{j\ge0}f^{(j)}(X)T^j.
\]
Fix an ordering $\alpha_1,\ldots,\alpha_n$ of $\mathbb F_q$,
with $n=q$. The univariate multiplicity code is
\[
\operatorname{Mult}_{s,\mathbb F_q,k}
 :=\left\{
   \left(
     \bigl(f^{(0)}(\alpha_i),\ldots,
           f^{(s-1)}(\alpha_i)\bigr)
   \right)_{i\in[n]}
   : f\in\mathbb F_q[X],\ \deg f<k
 \right\}
 \subseteq\Sigma^n.
\]
For $k<ns$, this is an $\mathbb F_q$-linear code of dimension $k$,
and its size is $q^k$.

Then $C_\lambda:=\operatorname{Mult}_{s,\mathbb F_q,k}$ is an
$(r,L_{\max})$-good family with
\[
r=\lfloor\lambda^2/2\rfloor,\qquad L_{\max}=2^\lambda.
\]
Moreover,
\[
\operatorname{dist}(C_\lambda^\perp)
\ge\frac{k+1}{s}>\lambda^2.
\]
\end{theorem}
\begin{proof}
The list-recovery property is~\cite[Corollary~4.18]{BHV:2026}.
For the diversity bound, the expansion result of~\cite{KTS:2022},
stated in~\cite[Theorem~4.16]{BHV:2026}, gives expansion
\[
A=q-\frac{k(s+1)}2(qK)^{1/(s+1)}>q/2,
\qquad K=2^{\lambda+1},
\]
for sufficiently large $\lambda$, as in the proof of
\cite[Corollary~4.17]{BHV:2026}.
Its neighbors are precisely the coordinate-symbol pairs for the
formal-derivative code. Formal derivatives are $j!$ times Hasse
derivatives, and these factors are nonzero because $q>s$.
This invertible change of symbols preserves all neighborhood sizes.
Thus $L\le2^\lambda$ distinct codewords touch at least $AL\ge nL/2$
coordinate-symbol pairs, proving~\eqref{tools:eq-code-diversity}.

The dual-distance bound follows from
\cite[Theorem~4.19 and Appendix~A]{BHV:2026}.
Theorem~A.5 and the proof of Corollary~5.23 there give efficient
unique decoding of the dual through invertible coordinate maps of
a multiplicity code. Since $2r\le\lambda^2<\operatorname{dist}(C_\lambda^\perp)$,
this decoder corrects every error of weight at most $r$.
Finally, $q$ can be found in polynomial time in $\lambda$, and
evaluating the first $s$ Hasse derivatives of
$1,X,\ldots,X^{k-1}$ at every $\alpha_i$ constructs a full-rank
generator matrix in polynomial time.
\end{proof}
\newpage\section{Code Intersection Problem}\label{tools:sec-code-intersection}

In this section, we define the biased code intersection problem and
give the tools that we need to study it. We first bound the
probability of guessing solutions, even when the algorithm receives
classical information about the hash. Then, we recall a quantum
state that solves the problem and define good hashes for which the
same state works for every target under consideration. Finally, we
show that biased hashes are good with overwhelming probability for
our choice of parameters.

\paragraph{Biased hash functions and code intersection.}
The biased variant of code intersection was introduced by
G{\"o}{\"o}s et al.~\cite{GGJL:2025} for the all-zero target.
We use its extension to low-Hamming-weight targets
from~\cite{BHV:2026}.
Consider a family of hash functions $H=(H_1,\ldots,H_n)$, where
$H_i:\Sigma\to\{0,1\}$. For $0<p\le1/2$, define the biased
distribution
\[
H\leftarrow\wBias_{n,p,\Sigma}
\]
by making all values $H_i(x)$, indexed by
$(i,x)\in[n]\times\Sigma$, mutually independent, with
$\operatorname{Pr}[H_i(x)=1]=p$.
For $c=(c_1,\ldots,c_n)\in\Sigma^n$, write
\[
H(c):=H_1(c_1)\|\cdots\|H_n(c_n).
\]

\begin{definition}[Code intersection relation]
\label{tools:def-code-intersection}
For a code $C\subseteq\Sigma^n$ and a hash family $H$, define
\[
R_{C,H}:=\{(b,c)\in\{0,1\}^n\times\Sigma^n:
                   c\in C,\ H(c)=b\}.
\]
The code intersection problem on a target $b\in\{0,1\}^n$ asks
for a codeword $c$ with $(b,c)\in R_{C,H}$.
Equivalently, it asks for an element of
$C\cap\prod_{i=1}^nH_i^{-1}(b_i)$.
This includes padded targets $b=z\|0^{n-\ell}$, for any integer
$0\le\ell\le n$ and $z\in\{0,1\}^\ell$.
\end{definition}

The following lemma bounds the probability of producing many
distinct solutions without querying the hash. The bound allows
the algorithm to receive a classical string that depends
arbitrarily on the hash.

\begin{lemma}[No-query guessing with classical side information]
\label{tools:lem-guessing-upper-bound}
For sufficiently large $\lambda$, fix a code $C=C_\lambda$ from
an $(r,L_{\max})$-good family, integers $1\le L\le L_{\max}$
and $t\ge0$, and $0<p\le1/2$.
Let $\Guesser$ be a randomized or quantum algorithm, fixed
independently of $H$, which makes no oracle queries and whose only
$H$-dependent input is a classical string $w\in\{0,1\}^t$.
For an output tuple $T=((b_j,c^{(j)}))_{j=1}^L$, let
$\mathsf{Valid}_H(T)$ mean that its pairs are distinct and
$(b_j,c^{(j)})\in R_{C,H}$ for every $j$. Then
\[
\mathbb E_{H\leftarrow\wBias_{n,p,\Sigma}}\!\left[
 \max_{w\in\{0,1\}^t}
 \Pr[T\leftarrow\Guesser(1^L,w)]{\mathsf{Valid}_H(T)}
\right]
\le\min\{1,\,2^t(1-p)^{nL/2}\}.
\]
In particular, the same upper bound holds for any choice of
$w=w(H)$. No running-time bound on $\Guesser$ is required.
For the parameters of \Cref{tools:thm-good-code-existence},
$p=\lambda^{-4}$ and $L\le2^\lambda$, the bound is at most
\[
\min\{1,\,2^t(1-\lambda^{-4})^{\lambda^5L/2}\}.
\]
\end{lemma}
\begin{proof}
Fix $w$ and an output tuple $T$ of $\Guesser(1^L,w)$.
Its distribution is independent of $H$. If $T$ can be valid, its
codewords must be distinct: a repeated codeword with different
targets cannot satisfy both equations. By~\eqref{tools:eq-code-diversity},
these codewords touch at least $nL/2$ distinct hash inputs $(i,x)$.
The targets prescribe one output bit at every such input.
Inconsistent prescriptions have probability zero; otherwise
independence and $p\le1/2$ give
\[
\Pr[H\leftarrow\wBias_{n,p,\Sigma}]{\mathsf{Valid}_H(T)}
\le(1-p)^{nL/2}.
\]
Averaging over the output distribution gives the same bound for
each fixed $w$. Writing
$a_w(H):=\Pr[T\leftarrow\Guesser(1^L,w)]{\mathsf{Valid}_H(T)}$,
we conclude that
\[
\mathbb E_H[\max_w a_w(H)]
\le\sum_{w\in\{0,1\}^t}\mathbb E_H[a_w(H)]
\le2^t(1-p)^{nL/2}.
\]
The left side is also at most one. The final specialization uses
$n\ge\lambda^5$.
\end{proof}

\subsection{The Code Intersection State and Good Hashes}
\label{tools:subsec-state-and-good-hashes}

We now turn to solving code intersection using a quantum state.
For this part, we only need efficient construction of the code and
error correction for its dual.

For the state and solver below, let $C\subseteq\Sigma^n$, where
$\Sigma=\mathbb F_q^s$, be an $\mathbb F_q$-linear code with an
efficiently computable generator matrix and efficient field
representation and arithmetic. Assume that $C^\perp$
admits efficient unique decoding of up to $r$ symbol errors.
The parameters $n,s,\log q$ are polynomially bounded in $\lambda$.
Fix $0<p\leq 1/4$ with $1/p=\operatorname{poly}(\lambda)$ and
an integer $0\leq\ell\leq n$.

\paragraph{The code intersection state.}
Liu~\cite{Liu:2023} observed that the target-independent
preprocessing of the Yamakawa--Zhandry algorithm can be supplied
as a quantum state, allowing code intersection to be solved
without querying $H$. We use the biased-hash formulation
of~\cite{BHV:2026}.

For $i\in[n]$ and $b\in\{0,1\}$, let
$S_{i,b}:=H_i^{-1}(b)$. When $S_{i,b}\neq\varnothing$,
define
\[
|H_i^{-1}(b)\rangle
 :=\frac{1}{\sqrt{|S_{i,b}|}}
   \sum_{x\in S_{i,b}}|x\rangle.
\]
Also define
\[
|\phi_2\rangle
 :=\frac{1}{\sqrt{|C|}}\sum_{c\in C}|c\rangle.
\]
If every $S_{i,b}$ is nonempty, set the \emph{code intersection state}
to be
\[
\ket{\psi_H}
 :=
 \left(\bigotimes_{i=1}^{n}|H_i^{-1}(0)\rangle\right)
 \otimes
 \left(\bigotimes_{i=1}^{n}|H_i^{-1}(1)\rangle\right)
 \otimes|\phi_2\rangle.
\]
Otherwise, set $\ket{\psi_H}=\ket{\bot}$,
where $\ket{\bot}$ is a distinguished failure state.
This convention makes the state well defined for every
hash family $H$. Its size is $O(ns\log q)=\operatorname{poly}(\lambda)$ qubits.

The Yamakawa--Zhandry algorithm~\cite{YZ:2022} exploits the
Fourier duality between a linear code and its dual.
The biased variant of~\cite{BHV:2026} accommodates codes
whose duals have subconstant relative distance by making
most target hash values zero. Write $\BiasedYZ(\ket{\psi_H},z)$
for this algorithm, using the decoder for $C^\perp$ and target
$z\|0^{n-\ell}$. The algorithm makes no queries to $H$.

We will use two properties of the hash: the solver should succeed
for every target, and both preimage sets at each coordinate should
have sufficient density. The latter property will be used for
cloning in \Cref{tools:sec-cloning}. We collect these requirements
in the following definition.

\begin{definition}[Good hashes]\label{tools:def-good-hash}
For the fixed code, decoder, and parameters above, let $\Good$
be the collection of all hash families $H$ satisfying:
\begin{enumerate}
    \item For every $z\in\{0,1\}^{\ell}$,
    \[
    \Pr[c\leftarrow\BiasedYZ(\ket{\psi_H},z)]{
        c\in C\ \land\ H(c)=z\|0^{n-\ell}}
    \geq 1-2^{-\lambda}.
    \]
    \item For every $i\in[n]$,
    \[
    \frac p2\leq\frac{|H_i^{-1}(1)|}{|\Sigma|}\leq 2p.
    \]
    In particular, $|H_i^{-1}(0)|/|\Sigma|\geq 1-2p\geq 1/2$.
\end{enumerate}
\end{definition}

\begin{lemma}[Probability of good hashes]\label{tools:lem-good-hash-probability}
Suppose $r\geq\ell$, put
\[
\mu:=\Pr[X\leftarrow\operatorname{Bin}(n-\ell,p)]{X>r-\ell},
\]
and assume
$16\mu^{1/4}+2^{-4\lambda}\leq2^{-\lambda}$. Then
\[
\Pr[H\leftarrow\wBias_{n,p,\Sigma}]{H\notin\Good}
\leq2^{\ell}\left(2n(1-p)^{|\Sigma|}+2\sqrt\mu\right)
+2n\exp\left(-p|\Sigma|/8\right).
\]
In particular, for the parameters of
\Cref{tools:thm-good-code-existence}, with
\[
p=\lambda^{-4},\qquad
\ell=\lceil\log_2 n\rceil+\lceil\log_2|\Sigma|\rceil
=\Theta(\lambda\log\lambda),
\]
for all sufficiently large $\lambda$,
\[
\Pr[H\leftarrow\wBias_{n,p,\Sigma}]{H\in\Good}
\geq1-2^{-2\lambda}.
\]
\end{lemma}

\begin{lemma}[Biased Yamakawa--Zhandry algorithm]\label{tools:lem-biased-yz}
There exists an efficient algorithm $\BiasedYZ$ such that, for
all sufficiently large $\lambda$, every $H\in\Good$, and every
$z\in\{0,1\}^{\ell}$,
\[
\Pr[c\leftarrow\BiasedYZ(\ket{\psi_H},z)]{
    c\in C\ \land\ H(c)=z\|0^{n-\ell}}
\geq1-2^{-\lambda}.
\]
The algorithm makes no oracle queries.
\end{lemma}

We prove these two lemmas in the following subsection.

\subsection{Solving Code Intersection and Probability of Good Hashes}

We first recall why the solver is efficient and correct for a fixed
good hash. Then, we bound the probability that a biased hash is good
and specialize the bound to the code parameters above.

\begin{proof}[Proof of \Cref{tools:lem-biased-yz}]
The algorithm runs in polynomial time by
\cite[Theorem~5.22]{BHV:2026} and the assumed efficiency of the
dual decoder. For every $H\in\Good$, its correctness for every
$z\in\{0,1\}^{\ell}$ is the first condition of
\Cref{tools:def-good-hash}.
\end{proof}

\begin{proof}[Proof of \Cref{tools:lem-good-hash-probability}]
We prove the two conditions in \Cref{tools:def-good-hash} in turn.
First, we consider correctness of the solver.
Put $\epsilon:=16\mu^{1/4}+2^{-4\lambda}$.
For a target $z\|0^{n-\ell}$, the error distribution in
\cite[Definition~5.15]{BHV:2026} has independent coordinates.
There are at most $\ell$ nonzero symbols in the prefix, and the
number of nonzero symbols in the suffix is distributed as
$\operatorname{Bin}(n-\ell,p)$. Thus the error has weight at most
$r$, and is corrected by the dual decoder, except with probability
at most $\mu$. Applying \cite[Theorem~5.22]{BHV:2026} gives, for
each fixed $z$, success probability at least $1-\epsilon$ except
on a set of hashes of probability at most
$2n(1-p)^{|\Sigma|}+2\sqrt\mu$. A union bound over the $2^\ell$
targets gives simultaneous correctness. Since
$\epsilon\leq2^{-\lambda}$, the first condition in
\Cref{tools:def-good-hash} fails with probability at most
$2^{\ell}(2n(1-p)^{|\Sigma|}+2\sqrt\mu)$.

We now turn to the density condition. Fix $i\in[n]$ and write
$X_i:=|H_i^{-1}(1)|$. Then
\[
X_i\sim\operatorname{Bin}(|\Sigma|,p),\qquad
m:=\mathbb E[X_i]=p|\Sigma|.
\]
By the multiplicative Chernoff bound,
\[
\Pr{X_i<m/2}\leq\exp(-m/8),\qquad
\Pr{X_i>2m}\leq\exp(-m/3).
\]
Hence $\Pr{X_i\notin[m/2,2m]}\leq2e^{-m/8}$.
A union bound over all $i\in[n]$ gives failure probability at
most $2n\exp(-p|\Sigma|/8)$. Dividing by $|\Sigma|$ gives the
required density bounds. Combining the two failure probabilities
proves the first assertion.

It remains to verify the bound for the concrete parameters.
Here $|\Sigma|=q^\lambda$ and
$q\geq\lambda^5$, so
\[
m=\frac{q^\lambda}{\lambda^4}\geq\lambda^{5\lambda-4}.
\]
Moreover, $n=q\leq2\lambda^5$. Thus the concentration failure
probability is at most $2^{-3\lambda}$ for all sufficiently
large $\lambda$.

To bound the solver error, let $a=\lfloor\lambda^2/4\rfloor$.
Since $r=\lfloor\lambda^2/2\rfloor$ and
$\ell=\Theta(\lambda\log\lambda)$, we have $r-\ell\geq a$
for sufficiently large $\lambda$. The mean of
$X\sim\operatorname{Bin}(n-\ell,p)$ is at most $2\lambda$.
The binomial tail bound therefore gives
\[
\mu\leq\Pr{X\geq a}
\leq\left(\frac{2e\lambda}{a}\right)^a
=2^{-\Omega(\lambda^2\log\lambda)}.
\]
It follows that $16\mu^{1/4}+2^{-4\lambda}\leq2^{-\lambda}$
and
\[
2^\ell\left(2n(1-p)^{|\Sigma|}+2\sqrt\mu\right)
\leq2^{-3\lambda}
\]
for sufficiently large $\lambda$. A union bound now gives
$\Pr[H\leftarrow\wBias_{n,p,\Sigma}]{H\notin\Good}
\leq2\cdot2^{-3\lambda}\leq2^{-2\lambda}$.
\end{proof}
\newpage\section{Self-Encrypted Oracles and Cloning}\label{tools:sec-cloning}

In this section, we define self-encrypted oracles and show that
they allow the code intersection state to be cloned. We first
state the cloning guarantee, and then give the algorithm and prove
its correctness.

Fix a code, decoder, and parameters as in
\Cref{tools:subsec-state-and-good-hashes}, and let $\ket{\psi_H}$
and $\Good$ be the state and good-hash family defined there.

\paragraph{Self-encrypted oracles.}
For cloning, assume $2^\ell\geq n|\Sigma|$, and fix an efficiently computable
injective encoding
$\bin:[n]\times\Sigma\longrightarrow\{0,1\}^{\ell}$.
Define $M[H]:[n]\times\Sigma\times\Sigma^n\longrightarrow
\{0,1,\bot\}$ by
\[
M[H](i,x,c)=
\begin{cases}
H_i(x) & \text{if $c\in C$ and
$H(c)=\bin(i,x)\|0^{n-\ell}$},\\
\bot & \text{otherwise}.
\end{cases}
\]
The oracle releases a hash value when supplied with a codeword
solving the corresponding code intersection problem.

\begin{lemma}[Correctness of cloning]\label{tools:lem-cloning}
There exists an efficient quantum algorithm $\Clone$ such that,
for every $H\in\Good$,
\[
F\left(\Clone^{M[H]}(\ket{\psi_H}\bra{\psi_H}),\ket{\psi_H}^{\otimes2}\right)
\geq 1-\operatorname{negl}(\lambda).
\]
\end{lemma}

We give the cloning algorithm and prove its correctness in the
following subsection.

\subsection{Cloning Correctness}

\begin{proof}[Proof of \Cref{tools:lem-cloning}]

The proof has two steps. First, we show how to obtain hash values
while approximately preserving the code intersection state. Then,
we repeat this procedure to prepare a new copy of the state.

For $i\in[n]$ and $x\in\Sigma$, let $U^{\YZ}_{i,x}$ be the
unitary purifying $\BiasedYZ(\cdot,\bin(i,x))$.
It takes a state register $\bA$ and zero-initialized working
register $\bWa$ to an output register $\bO$ encoding a vector in
$\Sigma^n$ and working register $\bWb$:
\[
U^{\YZ}_{i,x}:\bA\otimes\bWa\longrightarrow\bO\otimes\bWb.
\]
By the construction of $\BiasedYZ$, we can assume the sizes of
$\bWa$ and $\bWb$ are independent of $i$ and $x$. Define the
controlled unitary $U^{\YZ}$ by
\[
\ket{i,x}\otimes\ket{\psi,w}
\longmapsto\ket{i,x}\otimes U^{\YZ}_{i,x}\ket{\psi,w},
\]
and denote its inverse by $U^{\YZ,-1}$.

The idea is to exploit the ability to solve the code intersection
problem with $\ket{\psi_H}$: apply $U^{\YZ}$ to unlock access to
$M[H]$ (to get the hash value), and then apply $U^{\YZ,-1}$ to
uncompute and get back the original $\ket{\psi_H}$.

Write $\Hash_i^{M[H]}$ for the following subroutine of the cloning
algorithm. Prepare $\ket{i}_{\bI}$ and a uniform superposition
of $x\in\Sigma$ in $\bX$,
apply $U^{\YZ}$ to the input state and zero-initialized workspace,
and query $M[H]$ into a zero-initialized response register $\bB$ with distinct
encodings for $0,1,\bot$. Apply $U^{\YZ,-1}$ and retain the
state register $\bA$ together with $\bX\otimes\bB$.

We can prepare $\ket{\phi_2}$ on our own using the generator
matrix of $C$. We describe publicly known state preparations
ideally; implementing each to error $2^{-2\lambda}$ adds only
a negligible error over polynomially many calls.
It suffices to prepare $\ket{H_i^{-1}(0)}$ and
$\ket{H_i^{-1}(1)}$ for each $i\in[n]$. We obtain them by measuring
the output bit of multiple copies of
\[
\frac{1}{\sqrt{|\Sigma|}}\sum_{x\in\Sigma}\ket{x}\ket{H_i(x)}.
\]
The density bounds in \Cref{tools:def-good-hash} ensure that
both measurement outcomes occur with sufficient probability.
We use the same preserved copy of $\ket{\psi_H}$ to run the
extraction algorithm several times.

\begin{figure}[H]
\centering
\begin{algbox}{$\Clone^{M[H]}(\ket{\psi_H}\bra{\psi_H})$:}
\begin{enumerate}[label=\arabic*., leftmargin=1.5em]
    \item For $i=1,\ldots,n$, run $\Hash_i^{M[H]}$ iteratively
    $R=\lceil\lambda/p\rceil$ times using the same preserved
    state register; retain the $R$ output pairs for each $i$.
    \item Measure all response registers $\bB$. For every
    $i\in[n]$, retain one $\bX$ register with outcome $0$ and
    one with outcome $1$; otherwise, output $\bot$.
    \item Prepare $\ket{\phi_2}$ without making any query.
    \item Re-arrange and output the preserved state register
    together with the newly assembled state.
\end{enumerate}
\end{algbox}
\caption{The algorithm clones $\ket{\psi_H}$.}
\label{tools:fig-clone}
\end{figure}

For $i\in[n]$, define
\[
|\psi_i^H\rangle:=\frac{1}{\sqrt{|\Sigma|}}
\sum_{x\in\Sigma}|x\rangle_X|H_i(x)\rangle_B.
\]
We first show that one call of $\Hash_i^{M[H]}$ approximately
implements
\[
\ket{\psi_H}\longmapsto\ket{\psi_H}|\psi_i^H\rangle.
\]
Fix $i\in[n]$ and $x\in\Sigma$, and let
$z=\bin(i,x)$. Write
\[
|\xi_{i,x}\rangle:=U^{\YZ}_{i,x}
\left(\ket{\psi_H}_A|0\rangle_{W_1}\right).
\]
Let $\Pi_{i,x}$ denote the projector onto those outputs $c\in C$
satisfying $H(c)=z\|0^{n-\ell}$. Since $H\in\Good$,
\Cref{tools:def-good-hash} gives
\[
\|(I-\Pi_{i,x})|\xi_{i,x}\rangle\|^2\leq 2^{-\lambda}.
\]
Let $h=H_i(x)$. After querying $M[H]$ and subsequently applying
$U^{\YZ,-1}_{i,x}$, the real state is
\[
U^{\YZ,-1}_{i,x}\left(
\Pi_{i,x}|\xi_{i,x}\rangle|h\rangle_B+
(I-\Pi_{i,x})|\xi_{i,x}\rangle|\bot\rangle_B\right).
\]
On the other hand, the ideal state is
\[
\ket{\psi_H}_A|0\rangle_{W_1}|h\rangle_B
=U^{\YZ,-1}_{i,x}\left(
\Pi_{i,x}|\xi_{i,x}\rangle|h\rangle_B+
(I-\Pi_{i,x})|\xi_{i,x}\rangle|h\rangle_B\right).
\]
The Euclidean distance between the real and ideal states is
therefore at most
\[
\left\|(I-\Pi_{i,x})|\xi_{i,x}\rangle
\bigl(|\bot\rangle-|h\rangle\bigr)\right\|
\leq\sqrt2\,\|(I-\Pi_{i,x})|\xi_{i,x}\rangle\|
\leq\sqrt2\,2^{-\lambda/2}.
\]
Because the states corresponding to distinct $x$ are supported
on orthogonal values of the $X$ register, the same bound holds
when $x$ is placed in the uniform superposition. After tracing
out the working registers, monotonicity of trace distance thus gives
\begin{equation}\label{tools:eq-hash-extraction}
D\left(\Hash_i^{M[H]}(\ket{\psi_H}\bra{\psi_H}),
\ket{\psi_H}\bra{\psi_H}\otimes|\psi_i^H\rangle\langle\psi_i^H|\right)
\leq\epsilon_{\mathrm{hash}},
\qquad\epsilon_{\mathrm{hash}}:=\sqrt2\,2^{-\lambda/2}.
\end{equation}
We next run the extraction procedure $R=\lceil\lambda/p\rceil$
times for every $i\in[n]$. The total number of calls is $T=nR$.
Let $\rho_{\mathrm{pre}}$ denote the real joint state obtained
after these $T$ calls and before performing any measurement,
and define the corresponding ideal state
\[
\sigma_{\mathrm{pre}}:=\ket{\psi_H}\bra{\psi_H}\otimes
\bigotimes_{i=1}^n
\left(|\psi_i^H\rangle\langle\psi_i^H|\right)^{\otimes R}.
\]
The ideal extraction channel appends a fresh
$|\psi_i^H\rangle\langle\psi_i^H|$ and acts as the identity on
all existing registers. Applying \Cref{lem:hybrid-channel} to
\eqref{tools:eq-hash-extraction}, we obtain
\begin{equation}\label{tools:eq-extraction-error}
\delta_{\mathrm{ext}}:=D(\rho_{\mathrm{pre}},\sigma_{\mathrm{pre}})
\leq nR\epsilon_{\mathrm{hash}}=\operatorname{negl}(\lambda).
\end{equation}

It remains to analyze the reconstruction procedure. Measuring
the $B$ register of one copy of $|\psi_i^H\rangle$ in the ideal
state gives outcome $b\in\{0,1\}$ with probability
\[
q_{i,b}:=\frac{|H_i^{-1}(b)|}{|\Sigma|}.
\]
Conditioned on obtaining outcome $b$, the $X$ register is exactly
\[
|H_i^{-1}(b)\rangle
:=\frac{1}{\sqrt{|H_i^{-1}(b)|}}
\sum_{x:H_i(x)=b}|x\rangle.
\]
Since $H\in\Good$, the density condition in
\Cref{tools:def-good-hash} gives
\[
q_{i,1}\geq p/2,\qquad q_{i,0}\geq 1/2.
\]
Consequently, among $R$ independent copies, the probability of
seeing no outcome $1$ is at most
\[
(1-q_{i,1})^R\leq e^{-Rq_{i,1}}\leq e^{-\lambda/2},
\]
whereas the probability of seeing no outcome $0$ is at most
$(1-q_{i,0})^R\leq 2^{-R}$. A union bound over all $i\in[n]$
therefore shows that reconstruction fails for at least one $i$
with probability at most
\begin{equation}\label{tools:eq-reconstruction-error}
\eta_{\mathrm{post}}:=n\left(e^{-\lambda/2}+2^{-R}\right)
=\operatorname{negl}(\lambda).
\end{equation}
Whenever reconstruction succeeds, we obtain one exact copy of
both $|H_i^{-1}(0)\rangle$ and $|H_i^{-1}(1)\rangle$ for every
$i\in[n]$. Since $|\phi_2\rangle$ can be prepared without oracle
queries, these registers form one fresh copy of $\ket{\psi_H}$,
while the original ideal state register remains untouched.
Thus, conditioned on success, the ideal reconstruction procedure
outputs exactly $\ket{\psi_H}^{\otimes2}$.

Regard the entire reconstruction procedure---including all
measurements, the choice of one register with each outcome, the
preparation of $|\phi_2\rangle$, and the failure output---as a
single CPTP map $\mathcal M$. Let
$\rho_{\mathrm{out}}:=\mathcal M(\rho_{\mathrm{pre}})$ and
$\sigma_{\mathrm{out}}:=\mathcal M(\sigma_{\mathrm{pre}})$.
By contractivity of trace distance and
\eqref{tools:eq-extraction-error},
\[
D(\rho_{\mathrm{out}},\sigma_{\mathrm{out}})\leq\delta_{\mathrm{ext}}.
\]
Let $\Pi_{\psi}:=\ket{\psi_H}^{\otimes2}\bra{\psi_H}^{\otimes2}$.
By \eqref{tools:eq-reconstruction-error},
$\Tr(\Pi_{\psi}\sigma_{\mathrm{out}})\geq1-\eta_{\mathrm{post}}$,
and hence
\[
\Tr(\Pi_{\psi}\rho_{\mathrm{out}})
\geq1-\eta_{\mathrm{post}}-\delta_{\mathrm{ext}}.
\]
For sufficiently large $\lambda$, therefore,
\[
F(\rho_{\mathrm{out}},\ket{\psi_H}^{\otimes2})
=\sqrt{\Tr(\Pi_{\psi}\rho_{\mathrm{out}})}
\geq\sqrt{1-\eta_{\mathrm{post}}-\delta_{\mathrm{ext}}}
=1-\operatorname{negl}(\lambda).
\]
This proves the lemma.

\end{proof}

\newpage\part{Separations}
\section{Separating $\cQMA$ from $\QCMA$}\label{sec:QMA}
In this section, we prove the following theorem.
\mainseparationQMA*

To this end, we prove an oracle-input separation between these two classes, through the following problem.

\begin{definition}[Self-Encrypted Code Intersection Subset Size Problem]
\label{subsec:CCISS}
Let $\Sigma$ be a finite alphabet of code symbols, $C\subseteq\Sigma^n$
a code of length $n$, and $0\le\ell\le n$. Fix an efficiently
computable injective encoding
$\bin:[n]\times\Sigma\hookrightarrow\{0,1\}^{\ell}$ and a family
$\mathcal G$ of binary hash functions
$H=(H_i:\Sigma\to\{0,1\})_{i\in[n]}$.
For $c\in\Sigma^n$, write $H(c):=H_1(c_1)\cdots H_n(c_n)$.
For any binary hash family $H$ and
$E\subseteq\{0,1\}^{\ell}\times\Sigma^n$, define the oracle
$M[H,E]:[n]\times\Sigma\times\Sigma^n\to\{0,1,\bot\}$ by
\[
M[H,E](i,x,c)=\begin{cases}
    H_i(x) & \text{if $c\in C$, $H(c)=\bin(i,x)\|0^{n-\ell}$,} \\
           & \quad\text{and $(\bin(i,x),c)\in E$,} \\
    \bot & \text{otherwise}.
\end{cases}
\]
The \emph{self-encrypted code intersection subset size (CISS)
problem} is to distinguish oracle-inputs $M[H,E]$ with:\footnote{For
$\mathcal G=\Good$ below, every encoded target has a codeword
preimage, so the YES and NO oracle families are disjoint.}
\begin{enumerate}
    \item \textbf{YES:} $H\in\mathcal G$ and
    $E=\{0,1\}^{\ell}\times\Sigma^n$.
    \item \textbf{NO:} $H\in\mathcal G$ and
    $E\subseteq F\times\Sigma^n$ for some
    $F\subseteq\{0,1\}^{\ell}$ with $|F|\le n|\Sigma|/3$.
\end{enumerate}
Write $M[H]:=M[H,\{0,1\}^{\ell}\times\Sigma^n]$ for the
self-encrypted oracle of \Cref{tools:sec-cloning}.
\end{definition}

\begin{theorem}[Oracle-input separation]
\label{thm:self-encrypted-ciss}
Let $(C_\lambda\subseteq\Sigma^n)_\lambda$ be an
$(r,L_{\max})$-good code family as in \Cref{tools:def-good-code}, with
\[
\begin{gathered}
q\text{ prime},\quad \lambda^5\le q\le2\lambda^5,\quad
n=q,\quad s=\lambda,\quad \Sigma=\mathbb F_q^s,\\
r=\lfloor\lambda^2/2\rfloor,\quad L_{\max}=2^\lambda,\quad
p=\lambda^{-4},\quad
\ell=\lceil\log_2 n\rceil+\lceil\log_2|\Sigma|\rceil.
\end{gathered}
\]
Set $\mathcal G=\Good$ from \Cref{tools:def-good-hash}, using the
dual decoder of the code family.\footnote{We restrict the promise
to sufficiently large $\lambda$ so that the preceding code and
hash lemmas apply; in particular, $\ell\le n$ and $p\le1/4$.
All algorithms receive $1^\lambda$ and know the public code and
encoding. We leave $1^\lambda$ implicit below.}
For the self-encrypted CISS problem, let
$\mathcal Y_\lambda,\mathcal N_\lambda$ be its YES and NO oracle families.
Then:
\begin{enumerate}
    \item The problem belongs to $\cQMA$ in the oracle-input model.
    \item For every pair of polynomial bounds $t,Q$, all sufficiently
    large $\lambda$, and every quantum verifier $\A$ with a
    $t(\lambda)$-bit classical witness and at most $Q(\lambda)$
    oracle queries, there exists
    $M\in\mathcal Y_\lambda\cup\mathcal N_\lambda$ such that
    \[
    \begin{aligned}
    &M\in\mathcal Y_\lambda
    &&\text{and}\quad
    \max_{w\in\{0,1\}^{t(\lambda)}}
    \Pr{\A^M(1^\lambda,w)=1}<\frac23,\\
    &\text{or}\quad M\in\mathcal N_\lambda
    &&\text{and}\quad
    \max_{w\in\{0,1\}^{t(\lambda)}}
    \Pr{\A^M(1^\lambda,w)=1}>\frac13.
    \end{aligned}
    \]
    In particular, the problem does not belong to $\QCMA$ in this model.
\end{enumerate}
\end{theorem}

Such code families exist by \Cref{tools:thm-good-code-existence}.
We prove the two parts in \Cref{subsec:cqma-proof-system,subsec:qcma-hardness},
respectively; the fixed-oracle separation then follows by the
diagonalization described in \Cref{subsec:complexity-theory}.
Throughout this section, write $\Bias=\wBias_{n,p,\Sigma}$.

\subsection{The $\cQMA$ Proof System}
\label{subsec:cqma-proof-system}

In this subsection, we prove the $\cQMA$ part of
\Cref{thm:self-encrypted-ciss}. On a YES instance, we use the code intersection state
from \Cref{tools:subsec-state-and-good-hashes} as our honest witness:
\[
\advstate:=\ket{\psi_H}
=\left(\bigotimes_{i=1}^n\ket{H_i^{-1}(0)}\right)
\otimes\left(\bigotimes_{i=1}^n\ket{H_i^{-1}(1)}\right)
\otimes\frac{1}{\sqrt{|C|}}\sum_{c\in C}\ket{c}.
\]
Here $\ket{H_i^{-1}(b)}$ is the uniform superposition over
$H_i^{-1}(b)$. The witness has $O(ns\log q)=\operatorname{poly}(\lambda)$
qubits.

We define the verifier $V$ as follows.\footnote{We use rejection
sampling with a cutoff of $\lambda$ attempts. Let $N=n|\Sigma|$.
In each attempt, sample a $\lceil\log_2 N\rceil$-bit integer and
retain it if it is less than $N$, interpreting it as a pair in
$[n]\times\Sigma$ using a fixed efficient enumeration. Each attempt
succeeds with probability at least $1/2$, so sampling fails with
probability at most $2^{-\lambda}$. Conditioned on success, the
pair is exactly uniform.}

\begin{algbox}{$V^{M[H,E]}(\rho)$}
\begin{enumerate}[label=\arabic*., leftmargin=1.5em]
    \item Sample $(i,x)$ uniformly from $[n]\times\Sigma$.
    If sampling fails, output $0$.
    \item Run $\BiasedYZ(\rho,\bin(i,x))$ and measure its output
    to obtain $c$. If $c\notin\Sigma^n$, output $0$.
    \item Query $M[H,E](i,x,c)$. Output $1$ if the answer is
    $0$ or $1$, and output $0$ if the answer is $\bot$.
\end{enumerate}
\end{algbox}

By \Cref{tools:lem-biased-yz}, the algorithm $\BiasedYZ$ runs in
polynomial time and makes no oracle queries. The sampling step
also runs in polynomial time. Thus, $V$ is a polynomial-time
uniform quantum query algorithm making at most one oracle query.

We first prove completeness. Fix $H\in\Good$ and let
$E=\{0,1\}^{\ell}\times\Sigma^n$. Conditioned on successful
sampling, consider any pair $(i,x)$ chosen by the verifier.
By \Cref{tools:lem-biased-yz}, on input $\advstate$ the solver
outputs a codeword $c\in C$ satisfying
\[
H(c)=\bin(i,x)\|0^{n-\ell}
\]
with probability at least $1-2^{-\lambda}$. Whenever this happens,
the oracle returns $H_i(x)\in\{0,1\}$, and the verifier accepts.
Since sampling succeeds with probability at least $1-2^{-\lambda}$,
we get
\[
\Pr{V^{M[H,E]}(\advstate)=1}
\ge (1-2^{-\lambda})^2\ge\frac23
\]
for all sufficiently large $\lambda$.

We now prove soundness. Fix any $H$, any
$E\subseteq F\times\Sigma^n$ with $|F|\le n|\Sigma|/3$, and any
witness $\rho$. Observe that, by the definition of $M[H,E]$,
the verifier can accept only if $\bin(i,x)\in F$, regardless of
the codeword produced by the solver. The pair $(i,x)$ is chosen
independently of the witness and, conditioned on successful
sampling, is uniform over $[n]\times\Sigma$. Since $\bin$ is
injective and a sampling failure causes rejection, it follows that
\[
\Pr{V^{M[H,E]}(\rho)=1}
\le\frac{\bigl|F\cap\bin([n]\times\Sigma)\bigr|}{n|\Sigma|}
\le\frac{|F|}{n|\Sigma|}\le\frac13.
\]
This proves soundness.

Finally, on a YES instance we have $M[H,E]=M[H]$. By
\Cref{tools:lem-cloning}, the same honest witness $\advstate$
can be efficiently cloned using this oracle with negligible error.
Together with completeness and soundness, this gives the claimed
$\cQMA$ proof system.

\subsection{The Non-Existence of $\QCMA$ Proof System}
\label{subsec:qcma-hardness}

In this subsection, we prove the $\QCMA$ lower bound in
\Cref{thm:self-encrypted-ciss}. We follow the classical-hardness
argument of \cite{BHV:2026}. At a high level, a verifier for the
self-encrypted CISS problem gives a query-free algorithm that
outputs many valid code-intersection pairs. We then apply
\Cref{tools:lem-guessing-upper-bound} to rule out such an algorithm.

Our oracle introduces one additional difficulty. In
\cite{BHV:2026}, a successful oracle query only reveals whether
the queried code-intersection pair belongs to the hidden set.
In our construction, a valid query $(i,x,c)$ additionally returns
$H_i(x)$. Thus, to simulate the oracle, the guesser must also
predict the hash values revealed by these queries. We will analyze
the guesser for each fixed value of its initial witness, and then
apply \Cref{tools:lem-guessing-upper-bound}. For conciseness, write
$\Guesser(1^L,w)$ for the algorithm in \Cref{fig:guesser}
conditioned on its initial witness being $w$.

For a tuple $T=((z_j,c^{(j)}))_{j=1}^L$, with
$z_j\in\{0,1\}^{\ell}$, write $\mathsf{Valid}_H(T)$ for the event
\[
\begin{gathered}
(z_j,c^{(j)})\neq(z_k,c^{(k)})\quad\text{for all }j\neq k,\\
c^{(j)}\in C,\qquad H(c^{(j)})=z_j\|0^{n-\ell}
\quad\text{for all }j\in[L].
\end{gathered}
\]
This is the validity condition of
\Cref{tools:lem-guessing-upper-bound} after padding each target.\footnote{For
set outputs, we order the pairs canonically and treat an output
with fewer than $L$ pairs as a failure.}

\begin{lemma}[Guessers from $\QCMA$ verifiers]
\label{lem:qcma-guesser}
Fix $\lambda$, and suppose that a quantum verifier $\A$ takes a
$t$-bit classical witness, makes at most $Q\ge1$ oracle queries,
and satisfies completeness and soundness for every promised
self-encrypted CISS instance at parameter $\lambda$.
For every integer $1\le L\le\lfloor n|\Sigma|/3\rfloor$, there
exists a query-free algorithm $\Guesser(1^L)$, fixed independently
of $H$, such that for every $H\in\Good$, some
$w_H\in\{0,1\}^t$ satisfies
\[
\Pr[T\leftarrow\Guesser(1^L,w_H)]{\mathsf{Valid}_H(T)}
\ge\left(\frac{1}{72Q^2}\right)^L.
\]
\end{lemma}

\begin{proof}
We use the guesser given in \Cref{fig:guesser}. First, note that
by purifying the verifier and padding it with unused queries if
necessary, we may write its execution as
\[
V_Q M V_{Q-1}M\cdots M V_0\ket{w,0},
\]
followed by a measurement. Here the $V_j$ are oracle-independent
unitaries, and $M$ denotes the oracle unitary with distinct binary
encodings for $0,1,\bot$.\footnote{For an unused query, put the
response register in the uniform superposition over all binary
basis states, which is invariant under every XOR. In
\Cref{fig:guesser}, $f_{\Delta_{r-1}}$ suppresses the argument
$\widehat H_{r-1}$, and the query prefix for $j=0$ is $V_0\ket{w,0}$.}

The guesser in \Cref{fig:guesser} maintains a set $\Delta$ of
code-intersection pairs and a table $\widehat H$ of guessed hash
bits. It simulates the verifier using these tables, measures a
random query, and adds the resulting pair to $\Delta$. All oracle
calls in this simulation are implemented from the tables, so the
guesser makes no queries to $H$ or $M[H,E]$.

\begin{figure}[H]
\centering
\begin{algbox}{$\Guesser(1^L)$:}
\begin{enumerate}[label=\arabic*., leftmargin=1.5em]
    \item Sample a random $w \in \bit^{t}$ and initialize $\Delta_0 = \emptyset$ and $\widehat{H}_0 = \emptyset.$
    \item For all $r \in [L]:$
        \begin{enumerate}
            \item Sample a random $j \gets \set{0, \ldots, Q-1}$.
            \item Compute the state $ V_j\widetilde{M}_{r-1}V_{j-1} \ldots \widetilde{M}_{r-1}V_0\ket{w, 0},$ where $\widetilde{M}_{r-1}$ is the oracle unitary defined by
            \[
                \ket{i, x, c}\ket{y}\ket{z} \mapsto \ket{i, x, c}\ket{y  \oplus f_{\Delta_{r-1}}(i,x,c)}\ket{z},
            \]
            where 
            \[
                f_{\Delta_{r-1}, \widehat{H}_{r-1}}(i, x, c) = \begin{cases}
                    \widehat{H}_{r-1}(i, x) & \text{if $(\bin(i,x)\|0^{n-\ell}, c)\in \Delta_{r-1}$}  \\
                    \bot & \text{otherwise}.
                \end{cases}
            \]
            \item Measure the first register in the standard basis for output $(i, x, c)$.
            \item Sample a hash value $b \gets \bit$ and set $\widehat H_r
            \coloneqq
            \widehat H_{r-1}[(i,x)\mapsto b]$.
            \item Update $\Delta_r \coloneqq \Delta_{r-1} \cup \set{(\bin(i,x)\|0^{n-\ell}, c)}$.
        \end{enumerate} 
    \item Outputs $\set{(z[1:\ell], c): (z, c) \in \Delta_L}.$
\end{enumerate}
\end{algbox}
\caption{}
\label{fig:guesser}
\label{FigHash}
\end{figure}

Fix any $H\in\Good$ for the rest of the proof. By completeness,
there is a witness $w_H$ such that $\A^{M[H]}$ accepts with
probability at least $2/3$. We condition on the initial witness
being $w_H$. We will show that, in each round, the guesser finds
a new valid pair and guesses its hash bit correctly with
probability at least $1/(72Q^2)$, conditioned on this having
happened in all preceding rounds.

For round $r$, assume that the first $r-1$ rounds have succeeded.
Note that the set $\Delta_{r-1}$ in the figure stores padded
targets. Thus, let
\[
E_{r-1}:=\{(z[1:\ell],c):(z,c)\in\Delta_{r-1}\}.
\]
Since all pairs in $\Delta_{r-1}$ are valid and all guessed hash
bits are correct, the simulated oracle is exactly $M[H,E_{r-1}]$.
Let $F_{r-1}$ be the projection of $E_{r-1}$ onto its first
coordinate. Observe that
\[
E_{r-1}\subseteq F_{r-1}\times\Sigma^n,
\qquad |F_{r-1}|\le r-1<L\le n|\Sigma|/3.
\]
Hence, $M[H,E_{r-1}]$ is a NO instance. By soundness,
$\A^{M[H,E_{r-1}]}$ accepts $w_H$ with probability at most $1/3$.

We now claim that measuring a uniformly chosen query in this
execution gives a new valid pair with probability at least
$1/(36Q^2)$. To see this, note that $M[H]$ and $M[H,E_{r-1}]$
differ precisely on the inputs in
\[
S_{r-1}:=\{(i,x,c):c\in C,\ H(c)=\bin(i,x)\|0^{n-\ell},
                  (\bin(i,x),c)\notin E_{r-1}\}.
\]
For $j\in\{0,\ldots,Q-1\}$, let $\mu_j$ be the query weight
on $S_{r-1}$ immediately before query $j+1$ in the
$M[H,E_{r-1}]$ execution. The acceptance probabilities of the
two executions differ by at least $1/3$. Thus, by
\Cref{lem:query-hybrid},
\[
\frac13\le2\left(Q\sum_{j=0}^{Q-1}\mu_j\right)^{1/2},
\qquad
\frac1Q\sum_{j=0}^{Q-1}\mu_j\ge\frac{1}{36Q^2},
\]
as claimed.

After measuring the query, the guesser samples a fresh bit for
$H_i(x)$, which is correct with probability $1/2$. This is true
also when the table already has an entry for $(i,x)$ and the new
sample replaces it. Therefore, conditioned on the first $r-1$
rounds having succeeded, round $r$ succeeds with probability at
least $1/(72Q^2)$. Multiplying these probabilities over all $L$
rounds completes the proof.
\end{proof}

\begin{lemma}
    For any constant $a > 0$, let $Q(\secp), t(\secp)$ be functions that satisfy $Q(\secp) \le a \secp^a, t(\secp) \le a \secp^a$ for sufficiently large $\secp$. Then for sufficiently large $\secp$, for all quantum query algorithms $\A$ which take a classical $t(\secp)$-bit witness and makes $Q(\secp)$ oracle queries $M[H,E]$ of size $\secp$, there exists an oracle $M[H^*,E^*]$ of size $\secp$ such that $H^* \in \Good$ and 
    \begin{enumerate}
        \item either $E^* = \bit^\ell \times \Sigma^n$, but for all $t(\secp)$-bit witnesses $w$,
        \[
            \Pr{\A^{M[H^*,E^*]}(w) = 1} < \frac{2}{3},
        \]
        \item or $E^* \subseteq F^*  \times \Sigma^n \subseteq \bit^\ell \times \Sigma^n$ such that $|F^*|\le n|\Sigma|/3$, but there exists a $t(\secp)$-bit witness $w$ such that 
        \[
            \Pr{\A^{M[H^*,E^*]}(w) = 1} > \frac{1}{3}.
        \]
    \end{enumerate} 
\end{lemma}

\begin{proof}
Suppose for a contradiction that $\A$ satisfies completeness and
soundness for every promised instance at parameter $\lambda$.
We write $t=t(\lambda)$ and $Q=\max\{1,Q(\lambda)\}$, and set
\[
L:=\max\{t,\lambda\}.
\]
Note that $L$ is polynomially bounded. Thus, for sufficiently
large $\lambda$, we have
\[
L\le\min\{2^\lambda,\lfloor n|\Sigma|/3\rfloor\}.
\]
We now apply \Cref{lem:qcma-guesser}. Let
\[
P_H(w):=\Pr[T\leftarrow\Guesser(1^L,w)]{\mathsf{Valid}_H(T)}.
\]
For every $H\in\Good$, the lemma gives
$\max_w P_H(w)\ge(72Q^2)^{-L}$. By
\Cref{tools:lem-good-hash-probability}, a hash sampled from $\Bias$
belongs to $\Good$ with probability at least $1-2^{-2\lambda}$.
Averaging over $H$, we obtain
\[
\mathbb E_{H\leftarrow\Bias}\!\left[\max_{w\in\{0,1\}^t}P_H(w)\right]
\ge(1-2^{-2\lambda})(72Q^2)^{-L}.
\]

On the other hand, note that for each fixed $w$, the distribution
of $\Guesser(1^L,w)$ is independent of $H$.
After padding each output target $z$ to $z\|0^{n-\ell}$,
\Cref{tools:lem-guessing-upper-bound} applies directly and gives
\begin{align*}
\mathbb E_{H\leftarrow\Bias}\!\left[\max_{w\in\{0,1\}^t}P_H(w)\right]
&\le2^t(1-\lambda^{-4})^{nL/2}\\
&\le2^t e^{-\lambda L/2}
\le\left(2e^{-\lambda/2}\right)^L,
\end{align*}
where we used $n\ge\lambda^5$ and $t\le L$.
Combining the two bounds and taking $L$th roots yields
\[
\frac{1}{144Q^2}
\le\frac{(1-2^{-2\lambda})^{1/L}}{72Q^2}
\le2e^{-\lambda/2}
\]
for sufficiently large $\lambda$. However, $Q$ is polynomially
bounded, so this is a contradiction. This completes the proof of
the lemma, and hence the
$\QCMA$ part of \Cref{thm:self-encrypted-ciss}.
\end{proof}
\newpage\section{Separating $\cBQPqpoly$ from $\BQPpoly$}\label{sec:BQP}
In this section, we prove the following theorem.
\mainseparationBQP*

To this end, we prove an oracle-input separation between these two
classes. We use the same code family, parameters, and self-encrypted
oracle $M[H]$ as in \Cref{thm:self-encrypted-ciss}. Write
$\Bias=\wBias_{n,p,\Sigma}$ and
\[
D:=\bin([n]\times\Sigma),\qquad N:=|D|=n|\Sigma|.
\]
In particular, $2^\ell/4\le N\le2^\ell$ and
$\ell=\Theta(\lambda\log\lambda)$.

\begin{definition}[Self-encrypted function evaluation]
\label{def:self-encrypted-evaluation}
Let $C\subseteq\Sigma^n$ be a code over a finite alphabet $\Sigma$,
let $H=(H_i:\Sigma\to\{0,1\})_{i\in[n]}$ be binary hash functions,
and let $G:D\to\{0,1\}$. Extend $G$ to $\{0,1\}^{\ell}$ by setting
$G(z)=0$ for $z\notin D$. Define
\[
M'[H,G](i,x,c)=\begin{cases}
    G(\bin(i,x))\|H_i(x)
        & \text{if $c\in C$ and $H(c)=\bin(i,x)\|0^{n-\ell}$,}\\
    \bot & \text{otherwise}.
\end{cases}
\]
The self-encrypted function evaluation problem is to compute $G(z)$,
given $z\in\{0,1\}^{\ell}$ and oracle access to $M'[H,G]$.
\end{definition}

\begin{theorem}[Oracle-input separation]
\label{thm:advice-oracle-input}
Fix an $(r,L_{\max})$-good code family and the parameters of
\Cref{thm:self-encrypted-ciss}. Let $\cD_\lambda$ be the distribution
of $(G,M')$ obtained by sampling $H\gets\Bias$ and, independently,
a uniformly random $G:D\to\{0,1\}$, and setting $M'=M'[H,G]$.
Then:
\begin{enumerate}
    \item There are polynomial-time uniform quantum algorithms
    $\A,\operatorname{Clone}$ and a family of
    $\operatorname{poly}(\lambda)$-qubit pure advice states
    $\{\ket{\alpha_{M'}}\}_{M'}$, depending only on $M'$, such that
    $\A$ makes at most one oracle query and
    \[
    \Pr[(G,M')\gets\cD_\lambda][\Bigg]{
    \begin{gathered}
    \forall z\in\{0,1\}^{\ell},\quad
    \Pr{\A^{M'}(\ket{\alpha_{M'}},z)=G(z)}\ge1-2^{-\lambda},\\
    F\!\left(\operatorname{Clone}^{M'}(\ket{\alpha_{M'}}),
             \ket{\alpha_{M'}}^{\otimes2}\right)
        \ge1-\negl(\lambda)
    \end{gathered}
    }\ge1-2^{-2\lambda},
    \]
    for all sufficiently large $\lambda$.
    \item For every pair of polynomial bounds $t,Q$, all sufficiently
    large $\lambda$, every quantum query algorithm $\B$ making at most
    $Q(\lambda)$ queries, and every family of $t(\lambda)$-bit classical
    advice strings $\{adv_{M'}\}_{M'}$,
    \[
    \Pr[(G,M')\gets\cD_\lambda,\ z\gets D]{
        \B^{M'}(adv_{M'},z)=G(z)}\le\frac35.
    \]
\end{enumerate}
\end{theorem}

We prove the two parts in \Cref{subsec:clonable-advice,subsec:classical-advice},
respectively, and then obtain the fixed-oracle separation by
diagonalization.\footnote{As in \Cref{subsec:quantum-query-model},
the lower bounds allow arbitrary oracle-independent computation and
nonuniform algorithms. The upper-bound algorithms are uniform and
polynomial time. All algorithms receive $1^\lambda$, which we leave
implicit. We use the usual fixed-length binary encoding of $(i,x)$,
so membership in $D$ and decoding on $D$ are efficient.}

\subsection{The $\cBQPqpoly$ Algorithm}
\label{subsec:clonable-advice}

In this subsection, we prove the first part of
\Cref{thm:advice-oracle-input}. We use the code intersection state from
\Cref{tools:subsec-state-and-good-hashes} as our quantum advice:
\[
\ket{\alpha_{M'[H,G]}}:=\ket{\psi_H}
=\left(\bigotimes_{i=1}^n\ket{H_i^{-1}(0)}\right)
\otimes\left(\bigotimes_{i=1}^n\ket{H_i^{-1}(1)}\right)
\otimes\frac{1}{\sqrt{|C|}}\sum_{c\in C}\ket{c}.
\]
Here we use this state whenever $H\in\Good$, and assign an arbitrary
fixed basis state to any remaining oracle.\footnote{This assignment
depends only on $M'$. Indeed, if $H\in\Good$, every encoded target
has a codeword preimage. Thus, for every $(i,x)$, some answer of
$M'[H,G](i,x,\cdot)$ reveals $H_i(x)$, uniquely determining $H$.
Consequently, two representations of the same oracle cannot give
different good hashes.}
The advice has $O(ns\log q)=\operatorname{poly}(\lambda)$ qubits and
is the same for every input $z$.

On input $z\notin D$, the algorithm outputs $0$. Otherwise, it decodes
$z=\bin(i,x)$ and runs $\BiasedYZ(\ket{\psi_H},z)$ to obtain $c$.
If $c\notin\Sigma^n$, it outputs $0$; otherwise, it queries
$M'[H,G](i,x,c)$. If the answer is not $\bot$, it outputs its
first bit; otherwise, it outputs $0$.

Fix any $H\in\Good$ and $G$ for the rest of the argument. First,
note that the algorithm is always correct outside $D$. For every
$z\in D$, \Cref{tools:lem-biased-yz} gives
\[
\Pr[c\gets\BiasedYZ(\ket{\psi_H},z)]{
    c\in C\ \land\ H(c)=z\|0^{n-\ell}}\ge1-2^{-\lambda}.
\]
Whenever this happens, the first bit returned by $M'[H,G]$ is $G(z)$.
Thus, the algorithm has the claimed correctness for every $z$.

We now prove that the advice can be cloned. Observe that $M'[H,G]$
allows us to simulate $M[H]$ coherently: compute the answer of $M'$,
copy its hash-bit component, or $\bot$ if the answer is $\bot$, to the
response register, and uncompute the answer of $M'$. This implements
one query to $M[H]$ using two queries to $M'$. Applying this simulation
to the cloning algorithm of \Cref{tools:lem-cloning} gives
\[
F\!\left(\operatorname{Clone}^{M'[H,G]}(\ket{\psi_H}),
         \ket{\psi_H}^{\otimes2}\right)\ge1-\negl(\lambda).
\]
The simulation is exact, so it does not change the cloning error.
Finally, \Cref{tools:lem-good-hash-probability} gives
\[
\Pr[H\gets\Bias]{H\in\Good}\ge1-2^{-2\lambda}.
\]
This proves the first part of \Cref{thm:advice-oracle-input}.

\subsection{The Non-Existence of Classical Advice}
\label{subsec:classical-advice}

In this subsection, we prove the second part of
\Cref{thm:advice-oracle-input}. We begin by proving a search-like
separation between $\cBQPqpoly$ and $\BQPpoly$, following
\cite{BHV:2026}.

\begin{lemma}\label{lem:avg-seach-separation}
For the code family and parameters of \Cref{thm:self-encrypted-ciss},
the following hold:
\begin{enumerate}
    \item There is a polynomial-time uniform quantum query algorithm
    $\A$ and a family of $\operatorname{poly}(\lambda)$-qubit advice
    states $\{\ket{\alpha_M}\}_M$, depending only on $M$, such that
    \[
    \Pr[H\gets\Bias][\Bigg]{
        \begin{gathered}
        \forall z\in\{0,1\}^{\ell},\\
        \Pr[c\gets\A^{M[H]}(\ket{\alpha_{M[H]}},z)]{
            c\in C\ \land\ H(c)=z\|0^{n-\ell}}
        \ge1-2^{-\lambda}
        \end{gathered}}
    \ge1-2^{-2\lambda}.
    \]
    On the same event, the advice can be efficiently cloned using
    $M[H]$ with negligible error.
    \item For all quantum query algorithms $\B$ making
    $Q(\lambda)=\operatorname{poly}(\lambda)$ queries to $M[H]$ and
    all families of $t(\lambda)=\operatorname{poly}(\lambda)$-bit
    classical advice $\{adv_M\}_M$ depending only on $M$,
    \[
    \Pr[H\gets\Bias,\ z\gets\{0,1\}^{\ell}][\Big]{
        c\in C\ \land\ H(c)=z\|0^{n-\ell}
        \;\Big|\;c\gets\B^{M[H]}(adv_{M[H]},z)}
    \le\negl(\lambda).
    \]
    The bound is uniform over algorithms with the stated query and
    advice bounds.
\end{enumerate}
\end{lemma}

\begin{proof}
The first part follows from
\Cref{tools:lem-good-hash-probability,tools:lem-biased-yz,tools:lem-cloning}:
use $\ket{\psi_H}$ and run $\BiasedYZ$ without making any queries.
As above, for a good hash the advice depends only on $M[H]$.

We now proceed to prove the second part. Suppose for a contradiction
that there exist a polynomial $P(\lambda)$, a query adversary $\B$,
and a family of classical advice strings such that, for infinitely
many $\lambda$, the success probability in the second part is at least
$1/P(\lambda)$. Write $a_H:=adv_{M[H]}$ and let $t=t(\lambda)$
be its length. We may assume $Q=Q(\lambda)\ge1$ by padding with
unused queries.

We convert $\B$ into another query adversary $\A_1$ that only takes
$a_H$ as input and makes $Q$ queries to $M[H]$. It samples a uniformly
random $z\gets\{0,1\}^{\ell}$, runs $\B^{M[H]}(a_H,z)$ to obtain
$c$, and outputs $(z,c)$. By the definition of $\B$, we have
\begin{equation}
\Pr[H\gets\Bias][\Big]{
    c\in C\ \land\ H(c)=z\|0^{n-\ell}
    \;\Big|\;(z,c)\gets\A_1^{M[H]}(a_H)}
\ge\frac{1}{P(\lambda)}.
\label{eq:A_1-lower-bound}
\end{equation}

For each $H$ and $S\subseteq\{0,1\}^{\ell}\times\Sigma^n$, define
$\mu_{H,\bar S}$ to be the total query mass that
$\A_1^{M[H,S]}(a_H)$ places on points where $M[H,S]$ and $M[H]$
differ. Note that we use the restricted oracle execution and keep
the advice $a_H$ fixed. Let $\smallset$ be the event that there exists
a set $L_H\subseteq\{0,1\}^{\ell}\times\Sigma^n$ such that
\[
|L_H|\le2^{\lambda/2},\qquad
\mu_{H,\overline{L_H}}\le\frac{1}{256P(\lambda)^2Q}.
\]
We first show that, if the query mass of $\A_1$ is sufficiently
spread out, then we can convert $\A_1$ into a sampler that generates
many valid hash-code pairs. We use $\mathsf{Valid}_H$ as defined in
\Cref{subsec:qcma-hardness}.

\begin{lemma}\label{lem:spread-sampler}
There is a query-free algorithm $\Guesser(1^L,w)$, fixed independently
of $H$, such that whenever $\smallset$ does not occur, for every
integer $1\le L\le2^{\lambda/2}$,
\[
\Pr[T\gets\Guesser(1^L,a_H)]{\mathsf{Valid}_H(T)}
\ge\left(\frac{1}{512P(\lambda)^2Q^2}\right)^L.
\]
\end{lemma}
\begin{proof}
The algorithm is the guesser in \Cref{fig:guesser}, with $\A_1$ in
place of the verifier and with its initial classical string fixed to
the input $w$, as in \Cref{subsec:qcma-hardness}.

For each round $r$, let $X_r$ be the event that the measured query
$(i,x,c)$ gives a new valid pair
$(\bin(i,x)\|0^{n-\ell},c)\notin\Delta_{r-1}$, and let $E_r$ be
the event that the guessed bit is $H_i(x)$. Condition on
$X_1\land E_1\land\cdots\land X_{r-1}\land E_{r-1}$. Then every
stored pair is valid and every stored hash bit is correct, so the
simulated oracle is exactly $M[H,S_{r-1}]$, where $S_{r-1}$ is obtained
by removing the padding from the targets in $\Delta_{r-1}$.

Since $|S_{r-1}|\le r-1\le2^{\lambda/2}$ and $\smallset$ does not
occur, the total query mass on points where this oracle differs from
$M[H]$ exceeds $1/(256P(\lambda)^2Q)$. Measuring a uniformly random
one of the $Q$ queries therefore gives
\[
\Pr{X_r\mid X_1\land E_1\land\cdots\land X_{r-1}\land E_{r-1}}
\ge\frac{1}{256P(\lambda)^2Q^2}.
\]
The guesser then samples an independent uniform bit. Thus, the
conditional probability of $X_r\land E_r$ is at least
$1/(512P(\lambda)^2Q^2)$. Multiplying these conditional probabilities
over all $L$ rounds proves the claim.
\end{proof}

Let $\Lambda$ be the infinite set of parameters for which
\Cref{eq:A_1-lower-bound} holds. We continue to show that $\A_1$'s
queries are concentrated on very few points for all but finitely
many $\lambda\in\Lambda$.

\begin{lemma}\label{lem:advice-query-concentration}
There are only finitely many $\lambda\in\Lambda$ such that
\[
\Pr[H\gets\Bias]{\smallset}<1-\frac{1}{4P(\lambda)}.
\]
\end{lemma}
\begin{proof}
Suppose for a contradiction that this inequality holds for infinitely
many $\lambda\in\Lambda$. By \Cref{lem:spread-sampler}, for every
$L\le2^{\lambda/2}$,
\[
\mathbb E_{H\gets\Bias}
\Pr[T\gets\Guesser(1^L,a_H)]{\mathsf{Valid}_H(T)}
>\frac{1}{4P(\lambda)}
\left(\frac{1}{512P(\lambda)^2Q^2}\right)^L.
\]
On the other hand, \Cref{tools:lem-guessing-upper-bound}, with $t$ bits
of classical side information, bounds the same probability by
\[
2^t(1-p)^{nL/2}\le2^t e^{-\lambda L/2}.
\]
Take $L=\max\{t,\lambda\}$, which is at most $2^{\lambda/2}$ for all
sufficiently large $\lambda$. Taking $L$th roots, the preceding two
bounds would imply
\[
\frac{1}{1024P(\lambda)^2Q^2}<2e^{-\lambda/2}
\]
for infinitely many $\lambda$, where we used
$(4P(\lambda))^{-1/L}\ge1/2$. This is impossible, since the left-hand
side is inverse polynomial and the right-hand side is exponentially
small. This completes the proof.
\end{proof}

Having shown that $\smallset$ occurs with high probability, we now
construct a sampler that generates many valid hash-code pairs.
For each $H$ satisfying $\smallset$, fix any set $L_H^*$ witnessing
this event; otherwise, set $L_H^*:=\emptyset$. Define the punctured
oracle
\[
M_H^*(i,x,c):=M[H,L_H^*](i,x,c)=\begin{cases}
H_i(x)&\text{if $c\in C$, $H(c)=\bin(i,x)\|0^{n-\ell}$,}\\
      &\quad\text{and $(\bin(i,x),c)\in L_H^*$,}\\
\bot&\text{otherwise}.
\end{cases}
\]
Let $s_H$ be the probability that $\A_1^{M_H^*}(a_H)$ outputs a
valid pair. Whenever $\smallset$ occurs, \Cref{lem:query-hybrid} bounds
the change in this probability by
\[
2\sqrt{Q\mu_{H,\overline{L_H^*}}}\le\frac{1}{8P(\lambda)}.
\]
Thus, by \Cref{eq:A_1-lower-bound,lem:advice-query-concentration}, for
all sufficiently large $\lambda\in\Lambda$,
\[
\mathbb E_{H\gets\Bias}s_H
\ge\frac{1}{P(\lambda)}-\Pr[H\gets\Bias]{\neg\smallset}
   -\frac{1}{8P(\lambda)}
\ge\frac{1}{2P(\lambda)}.
\]

We define a sampler $\A_2$ which takes $a_H$ and $1^L$ as inputs and
runs $\A_1^{M_H^*}$ independently $L$ times. Fix $H$ for the moment.
Since $\A_1$ samples a fresh uniform target in each round, the
probability that its next target equals one of the preceding targets
is at most $(r-1)/2^\ell$, even conditioned on the preceding rounds
being successful. Consequently,
\[
\Pr[T\gets\A_2^{M_H^*}(a_H,1^L)]{\mathsf{Valid}_H(T)}
\ge\left(\max\left\{0,s_H-\frac{L}{2^\ell}\right\}\right)^L.
\]
We can now average over $H$. By convexity of
$u\mapsto(\max\{0,u-L/2^\ell\})^L$, for
$L=\lfloor2^{3\lambda/4}\rfloor$ and sufficiently large $\lambda$,
\begin{align*}
\mathbb E_{H\gets\Bias}
\Pr[T\gets\A_2^{M_H^*}(a_H,1^L)]{\mathsf{Valid}_H(T)}
&\ge\left(\mathbb E_Hs_H-\frac{L}{2^\ell}\right)^L\\
&\ge\left(\frac{1}{3P(\lambda)}\right)^L.
\end{align*}
Here we used $\ell=\Theta(\lambda\log\lambda)$.

Finally, define a query-free sampler $\A_3$ that receives $a_H$ and
a description of $M_H^*$ as classical input and simulates $\A_2$.
Since $\bin$ is injective, $M_H^*$ has at most $2^{\lambda/2}$
non-$\bot$ inputs. Its table therefore has a description of length
\[
d=O\!\left(2^{\lambda/2}(\ell+n\log|\Sigma|)\right)
\le2^{2\lambda/3}.
\]
Thus, $\A_3$ receives at most $t+d\le L$ bits of classical side
information for sufficiently large $\lambda$. Applying
\Cref{tools:lem-guessing-upper-bound} once more gives the contradicting
upper bound
\[
\mathbb E_{H\gets\Bias}
\Pr[T\gets\A_3(1^L,a_H,\operatorname{desc}(M_H^*))]{
    \mathsf{Valid}_H(T)}
\le2^{t+d}e^{-\lambda L/2}
\le\left(2e^{-\lambda/2}\right)^L
<\left(\frac{1}{3P(\lambda)}\right)^L.
\]
This proves the second part. The bounds used only $t,Q$, and $P$,
so the conclusion is uniform over the stated algorithms.
\end{proof}

We next recall a lemma from \cite{BHV:2026} that will be used to
pass from search to decision.
\begin{lemma}[Adapted from {\cite[Lemma~4.21]{BHV:2026}}]
\label{lem:yao-box}
Let $G:D\to\{0,1\}$ be a uniformly random function, where $N=|D|$
satisfies
\[
\left(\frac{(t(\lambda)+\log N)Q(\lambda)}{N}\right)^{1/3}
=\negl(\lambda).
\]
Let $\B$ be a quantum query algorithm making
$Q(\lambda)=\operatorname{poly}(\lambda)\ge1$ queries to $G$, and
let $\{adv_G\}_G$ be a family of
$t(\lambda)=\operatorname{poly}(\lambda)$-bit classical advice strings.
Suppose that
\[
\Pr[G\gets\{D\to\{0,1\}\},\ z\gets D]{
    \B^G(adv_G,z)=G(z)}>\frac35.
\]
Then, for all sufficiently large $\lambda$, for at least a
$1/(4000Q(\lambda)^2)$ fraction of $v\in D$, measuring a uniformly
random query of $\B^G(adv_G,v)$, for uniformly random $G$, produces
$v$ with probability at least $1/(3200Q(\lambda)^2)$.
\end{lemma}
\begin{proof}
For $z\in D$, let $G^{[z]}$ agree with $G$ except that
$G^{[z]}(z)=0$. The finite-domain Yao's box bound
\cite[Theorem~4.20]{BHV:2026} gives
\begin{equation}
\Pr[G,\ z\gets D]{\B^{G^{[z]}}(adv_G,z)=G(z)}
\le\frac12+O\!\left(\frac{(t+\log N)Q}{N}\right)^{1/3}
=\frac12+\negl(\lambda).
\label{eq:yao-punctured-bound}
\end{equation}
Let $m_{G,z}$ be the total query mass at $z$ in the execution with
$G$. The hybrid bound gives a difference in success probabilities
of at most $4\sqrt{Qm_{G,z}}$. Averaging and applying Cauchy--Schwarz,
we obtain
\[
\frac1{10}-\negl(\lambda)
\le4\sqrt{Q\mathbb E_{G,z}m_{G,z}},\qquad
\mathbb E_{G,z}\frac{m_{G,z}}Q
\ge\frac{1}{1600Q^2}-\negl(\lambda).
\]
Put $r_z:=\mathbb E_G[m_{G,z}/Q]\le1$. If $\theta$ is the fraction
of $z$ for which $r_z\ge1/(3200Q^2)$, then
$\mathbb E_zr_z\le\theta+1/(3200Q^2)$. It follows that
$\theta\ge1/(3200Q^2)-\negl(\lambda)\ge1/(4000Q^2)$ for all
sufficiently large $\lambda$, as claimed.
\end{proof}

\begin{proof}[Proof of the second part of \Cref{thm:advice-oracle-input}]
Suppose for a contradiction that there exist an algorithm $\B$
making $Q=\operatorname{poly}(\lambda)$ queries and a family of
$t=\operatorname{poly}(\lambda)$-bit advice strings such that, for
infinitely many $\lambda$,
\[
\Pr[H\gets\Bias,\ G\gets\{D\to\{0,1\}\},\ z\gets D]{
\B^{M'[H,G]}(adv_{M'[H,G]},z)=G(z)}>\frac35.
\]
We apply the argument of \Cref{lem:yao-box} directly to the valid
queries of $\B$. For a target $z$, replace $M'[H,G]$ by
$M'[H,G^{[z]}]$, keeping the original advice fixed. These oracles
differ only on valid queries $(i,x,c)$ with $\bin(i,x)=z$.

Fix $H$ for the moment. By hardwiring $H$, we can simulate the
modified oracle with at most two queries to $G^{[z]}$ per query of
$\B$.\footnote{Compute the queried value of $G^{[z]}$ in a scratch
register, use it only if the hardwired $H$ validates $(i,x,c)$,
and uncompute the scratch register. The Yao bound allows arbitrary
computation between queries.}
Thus, \Cref{eq:yao-punctured-bound} bounds the modified algorithm's
success, averaged over $G,z$, by $1/2+\negl(\lambda)$, uniformly in
$H$. The same bound holds after averaging over $H$.

Let $m_{H,G,z}$ be the total query mass of the original execution on
\[
\{(i,x,c):\bin(i,x)=z,\ c\in C,\ H(c)=z\|0^{n-\ell}\}.
\]
By the hybrid bound and Cauchy--Schwarz, for all sufficiently large
parameters under consideration,
\[
\frac1{20}\le4\sqrt{Q\mathbb E_{H,G,z}m_{H,G,z}},\qquad
\mathbb E_{H,G,z}\frac{m_{H,G,z}}Q\ge\frac{1}{6400Q^2}.
\]
In particular, measuring a uniformly random query of the original
algorithm produces a valid codeword preimage of its target with
inverse-polynomial probability.

We now obtain a contradiction to the second part of
\Cref{lem:avg-seach-separation}. For each fixed $G$, construct
$\B'[G]$ as follows. On input $z$ and advice $adv_{M'[H,G]}$, simulate
$\B^{M'[H,G]}$ using $M[H]$ and the hardwired $G$, stop immediately
before a uniformly random query of $\B$, measure $(i,x,c)$, and
output $c$. Each simulated query uses at most two queries to
$M[H]$, by computing its answer, attaching the appropriate $G$ bit
when the answer is not $\bot$, and uncomputing. Moreover, for fixed
$G$ the advice depends only on $M[H]$, since $M'[H,G]$ is determined
by $M[H]$ and $G$.

By averaging, there is a fixed $G^*$ for which
\[
\Pr[H\gets\Bias,\ z\gets D]{
c\in C\ \land\ H(c)=z\|0^{n-\ell}
\;\Big|\;c\gets\B'[G^*]^{M[H]}(adv_{M'[H,G^*]},z)}
\ge\frac{1}{6400Q^2}.
\]
When the target is instead uniform in $\{0,1\}^{\ell}$, let this
algorithm fail outside $D$. Since $|D|/2^\ell\ge1/4$, its success
is still at least $1/(25600Q^2)$, contradicting
\Cref{lem:avg-seach-separation}. This proves the second part of
\Cref{thm:advice-oracle-input}.
\end{proof}

The fixed-oracle separation follows by the diagonalization
described in \Cref{subsec:complexity-theory}, which also preserves
clonability.

\newpage\part{Applications}
\section{Applications to Quantum Fire in a Classical Oracle}
\label{sec:crypto-applications}
Our construction gives a good-code-based approach to quantum cryptography.
Its basic quantum resource is the code-intersection advice state
$\AdvState{H}$, together with classical oracles that control how the state
can be used. This differs from the hidden-subspace approach to quantum
money~\cite{AC:2012}, signature tokens~\cite{BS:2023}, and the oracle construction
of one-shot signatures~\cite{AGKZ:2020}; the latter construction's security
justification was subsequently corrected by~\cite{SZ:2025}. The classical-oracle
quantum-fire construction of~\cite{CGS:2026} also uses subspace-based one-shot
signature machinery.

Building on the Yamakawa--Zhandry algorithm~\cite{YZ:2022} and the good-code
framework of~\cite{BHV:2026}, our cloning and query-extraction techniques give
a direct proof of noninteractive untelegraphability for quantum fire.
This provides a simpler route to the classical-oracle existence result
of~\cite{CGS:2026}, whose stronger key-fire and interactive-security results
are not asserted here.
%We also describe a quantum-money mini-scheme and an honestly generated signature-token candidate; their security remains open.
Formal definitions and the quantum-fire proof appear in
Appendices~\ref{app:crypto-definitions} and~\ref{app:missing-security}.

\subsection{Classical Oracle}

\paragraph{Base oracles.}
Retain the parameters $n,\Sigma,C$, the biased hash distribution, and the
advice states from the preceding sections. Fix an injective binary encoding
$\bin:[n]\times\Sigma\to\{0,1\}^{\ell}$ and put
$\mathcal Z_\lambda:=\bin([n]\times\Sigma)$.
Uniform sampling from $\mathcal Z_\lambda$ means uniform sampling of $(i,a)$
followed by this encoding; unused binary encodings are not sampled.
For $H=(H_i)_{i\in[n]}$, recall the oracle
%\[
% O[H,E](z,c)=
% \begin{cases}
% 1,&c\in C,\ H(c)=z\Vert0^{n-\ell},\ (z,c)\in E,\\
% 0,&\text{otherwise},
% \end{cases}
%\]
%and
\[
 M[H](i,a,c)=
 \begin{cases}
 H_i(a),&c\in C,\ H(c)=\bin(i,a)\Vert0^{n-\ell}\\
 \bot,&\text{otherwise}.
 \end{cases}
\]
%Thus $O$ returns only verification outcomes, whereas $M$ also returns the hash bit needed by our cloning algorithm. Malformed inputs are rejected.

\paragraph{Global classical oracles.}
Let $I_\lambda=\{0,1\}^{\lambda}$. Independently sample
$H_u\leftarrow\wBias_{n,1/\lambda^4,\Sigma}$ for every $u\in I_\lambda$
and a uniformly random permutation $P:I_\lambda\to I_\lambda$.
Define
\[
 \mathbf H(u,i,a)=(H_u)_i(a),\qquad
 \mathbf M(P(u),i,a,c)=M[H_u](i,a,c)
\]
and set
\[
 \cM=(\mathbf H,\mathbf M,P)
\]
All components are classical function oracles with coherent evaluation
access. In particular, $P$ is accessed as
$\ket{u,v}\mapsto\ket{u,v\oplus P(u)}$; inverse evaluation is not supplied.
The inverse permutation occurs only in the mathematical definition of the
locked oracles. Generation knows $u$ and can prepare $\AdvState{H_u}$,
while the public serial number $P(u)$ permits access to the weaker locked
oracle. In~\cite[Section 8]{CGS:2026}, independently sampled classical oracles
serve as serial numbers for simplicity, and conversion to a single global
oracle with classical-string serial numbers is explicitly noted. Our
permutation implements this standard globalization step.

\subsection{Construction of Quantum Fire}
There exists a noninteractively untelegraphable quantum fire mini-scheme
relative to a classical oracle. We obtain the classical-oracle existence
result of~\cite{CGS:2026} using the following code-based construction.

\begin{construction}[Quantum fire mini-scheme]
\label{con:code-fire}
Relative to $\cM$, define the following algorithms.
\begin{enumerate}
\item $\QFireMini.\Spark(1^\lambda)$: Sample $u\leftarrow I_\lambda$,
prepare $\AdvState{H_u}$ using $\mathbf H(u,\cdot,\cdot)$, and output
\[
 (s,R)=(P(u),\AdvState{H_u}).
\]
Preparation uses a polynomial cutoff and
returns $\bot$ on failure.
\item $\QFireMini.\Clone(s,R)$: Apply our cloning algorithm with oracle
$\mathbf M(s,\cdot,\cdot,\cdot)$ and output the two flame registers.
\item $\QFireMini.\Ver(s,R)$: Apply the amplified verifier
$\widetilde V_M$ with oracle $\mathbf M(s,\cdot,\cdot,\cdot)$ and output
its acceptance bit. The underlying verifier samples $(i,a)$ uniformly,
runs $\BiasedYZ(R,\bin(i,a))$, and accepts precisely when the oracle
response is different from $\bot$.
\end{enumerate}
\end{construction}

\begin{theorem}[Correctness and untelegraphability]
\label{thm:code-fire}
Construction~\ref{con:code-fire} satisfies verification correctness and
strong cloning correctness.

More explicitly, let $A=(A_1,A_2)$ be a uniform query-bounded
adversary, with no shared quantum state between the stages, whose classical
transcript has at most $t$ bits. Let $Q_H,Q_P$ bound its total queries to
$\mathbf H,P$, respectively, and let $Q\ge1$ bound the queries to the
target locked oracle in $A_2$ followed by $\widetilde V_M$.
Put $T=Q_P+2Q_H+1$. For any integer
$1\le L\le\min\{2^\lambda,\lfloor|\mathcal Z_\lambda|/3\rfloor\}$,
\begin{align}
 \operationname{Pr}[A\text{ wins}]
 &\le b_\lambda+\eta_\lambda+\delta_\lambda
 +2^{t/(2L)}\sqrt{8}\,Q
       (1-\lambda^{-4})^{n/4}
 +2Q_H\sqrt{\operatorname{Inv}_\lambda(T)}.
 \label{eq:fire-winning-bound}
\end{align}
Here $b_\lambda=\operationname{Pr}[H\notin\Good]\le2^{-\lambda}$,
$\eta_\lambda$ bounds preparation failure conditioned on $H\in\Good$,
$\delta_\lambda\le2^{-\lambda}$ is amplified soundness error, and
$\operatorname{Inv}_\lambda(T)$ is the optimal success probability of
inverting a random $2^\lambda$-point permutation using at most $T$ forward
queries and auxiliary information independent of the permutation.
For preparation with $\lambda^6$ samples per coordinate,
$\eta_\lambda\le2n\exp(-\lambda^2/2)$.
For polynomial $t,Q,Q_H,Q_P$, choose $L=\max\{1,t\}$ to obtain
\[
 \operationname{Pr}[A\text{ wins}]
 \le b_\lambda+\eta_\lambda+\delta_\lambda
      +4Qe^{-\lambda/4}
      +O\!\left(\frac{Q_H(T+1)}{2^{\lambda/2}}\right)
 =\negl(\lambda).
\]
\end{theorem}
The proof is given in Appendix~\ref{app:missing-security}. The permutation
term uses the inversion bound (\cite{Amb:2000} for example). The theorem implies
the existence statement above and, with signatures secure in the same
adversary model, the standard public-key upgrade.

\bibliographystyle{alpha}
\bibliography{cmd/refs}

\newpage\part{Appendices}\appendix
\section{Definitions of Quantum Cryptography}
\label{app:crypto-definitions}
All honest algorithms are uniform quantum polynomial-time algorithms.
Classical strings are measured classical outputs. In the oracle setting,
every algorithm and adversary receives coherent evaluation access to the
specified global oracle. Query-bounded adversaries may use arbitrary
oracle-independent computation between polynomially many queries, but
receive no oracle-dependent nonuniform advice unless explicitly allowed.
Probability includes oracle sampling, honest randomness, and measurements.
An execution in which honest generation outputs $\bot$ is unsuccessful.

% ===== APPENDIX: DIGITAL SIGNATURES =====
\subsection{Digital Signatures}
\label{subsec:ds-definitions}
\begin{definition}[Post-quantum digital signatures]
A digital signature scheme $\DS=(\KeyGen,\Sign,\Ver)$ consists of:
\begin{itemize}
\item $\DS.\KeyGen(1^\lambda)$, which outputs classical keys $(vk,sk)$;
\item $\DS.\Sign(sk,m)$, which outputs a classical signature $\sigma$;
\item $\DS.\Ver(vk,m,\sigma)$, which outputs a bit.
\end{itemize}
For every admissible message $m$, correctness requires
\[
 \operationname{Pr}\!\left[
 \DS.\Ver(vk,m,\sigma)=1:
 (vk,sk)\leftarrow\DS.\KeyGen(1^\lambda),\
 \sigma\leftarrow\DS.\Sign(sk,m)
 \right]\ge1-\negl(\lambda).
\]
Existential unforgeability under adaptive chosen-message attacks requires
that a $\QPT$ adversary given $vk$ and classical-query access to
$\DS.\Sign(sk,\cdot)$ cannot output an accepting pair
$(m^*,\sigma^*)$ with $m^*$ never submitted to its signing oracle,
except with negligible probability. In a query-bounded oracle theorem,
the signature scheme must be secure against that same adversary model;
ordinary computational security alone does not cover unbounded
inter-query computation.
\end{definition}

% ===== APPENDIX: FIRE =====
\subsection{Quantum Fire}
\label{subsec:fire-definitions}
We follow the quantum-fire framework of~\cite{NZ:2024,BNZ:2025,CGS:2026} and
consider noninteractive untelegraphability without pre-shared entanglement.

Intuitively, a quantum fire state can be efficiently cloned while it remains
quantum, but cannot be compressed into polynomial-size classical information
from which a valid flame can later be reconstructed.

\begin{definition}[Public-Key Quantum Fire \cite{NZ:2024}]
A public-key quantum fire scheme $\mathsf{QFire}$ consists of the following
efficient quantum algorithms.

\begin{itemize}
    \item $\mathsf{QFire.Setup}(1^\lambda)$:
    On input the security parameter, output a public key $pk$ and a secret key
    $sk$.

    \item $\mathsf{QFire.Spark}(sk)$:
    On input the secret key, output a serial number $sn$ and a flame state in a
    register $R_{\mathrm{flame}}$.

    \item $\mathsf{QFire.Clone}(pk,sn,R_{\mathrm{flame}})$:
    On input the public key, a serial number, and a flame register, output the
    input flame register $R_{\mathrm{flame}}$ together with a new flame register
    $R'_{\mathrm{flame}}$.

    \item $\mathsf{QFire.Ver}(pk,sn,R)$:
    On input the public key, a serial number, and an alleged flame state,
    output $\mathsf{TRUE}$ or $\mathsf{FALSE}$.
\end{itemize}

We require \emph{verification correctness}:
\[
\operatorname{Pr}\left[
    \begin{array}{c}
        \mathsf{out}_{\mathrm{Ver}}=\mathsf{TRUE}
    \end{array}
    :
    \begin{array}{l}
        pk,sk\leftarrow\mathsf{QFire.Setup}(1^\lambda),\\
        sn,R_{\mathrm{flame}}\leftarrow
            \mathsf{QFire.Spark}(sk),\\
        \mathsf{out}_{\mathrm{Ver}}\leftarrow
            \mathsf{QFire.Ver}(pk,sn,R_{\mathrm{flame}})
    \end{array}
\right]
\geq 1-\operatorname{negl}(\lambda),
\]
and \emph{cloning correctness}:
\[
\operatorname{Pr}\left[
    \begin{array}{c}
        \mathsf{out}_{1,\mathrm{Ver}}=\mathsf{TRUE}
        \ \wedge\
        \mathsf{out}_{2,\mathrm{Ver}}=\mathsf{TRUE}
    \end{array}
    :
    \begin{array}{l}
        pk,sk\leftarrow\mathsf{QFire.Setup}(1^\lambda),\\
        sn,R_{\mathrm{flame}}\leftarrow
            \mathsf{QFire.Spark}(sk),\\
        R_{\mathrm{flame}},R'_{\mathrm{flame}}\leftarrow
            \mathsf{QFire.Clone}(pk,sn,R_{\mathrm{flame}}),\\
        \mathsf{out}_{1,\mathrm{Ver}}\leftarrow
            \mathsf{QFire.Ver}(pk,sn,R_{\mathrm{flame}}),\\
        \mathsf{out}_{2,\mathrm{Ver}}\leftarrow
            \mathsf{QFire.Ver}(pk,sn,R'_{\mathrm{flame}})
    \end{array}
\right]
\geq 1-\operatorname{negl}(\lambda).
\]

The output of $\mathsf{QFire.Clone}$ is a mixed state in general and $R_{\mathrm{flame}}$ and $R'_{\mathrm{flame}}$ can be entangled. We say that the scheme satisfies \emph{strong cloning correctness} if,
whenever
\[
sn,\ket{\psi_{sn}}\leftarrow\mathsf{QFire.Spark}(sk),
\]
the fidelity
\[
F(\mathsf{QFire.Clone}(\ket{\psi_{sn}}\bra{\psi_{sn}}),\ket{\psi_{sn}}^{\otimes 2})
\]
is $1-\negl(\lambda)$.
\end{definition}

\begin{definition}[Untelegraphability]
Let $\mathcal A=(\mathcal A_1,\mathcal A_2)$ be an adversary.
Consider the following game.

\begin{enumerate}
    \item The challenger samples
    \[
        pk,sk\leftarrow\mathsf{QFire.Setup}(1^\lambda)
    \]
    and sends $pk$ to $\mathcal A_1$.

    \item For polynomially many rounds, the challenger samples
    \[
        sn,R_{\mathrm{flame}}
        \leftarrow\mathsf{QFire.Spark}(sk)
    \]
    and sends $(sn,R_{\mathrm{flame}})$ to $\mathcal A_1$.

    \item $\mathcal A_1$ outputs a polynomial-size classical string $L$.

    \item $\mathcal A_2(L)$ outputs $sn^\star$ and a quantum register $R'$.

    \item The challenger outputs
    \[
        \mathsf{QFire.Ver}(pk,sn^\star,R').
    \]
\end{enumerate}

We say that $\mathsf{QFire}$ is \emph{untelegraphable} if every QPT
adversary has negligible probability of making the challenger output
$\mathsf{TRUE}$.

For unconditional oracle constructions, we instead quantify over
\emph{query-bounded} adversaries, which make polynomially many oracle
queries but may perform arbitrary computation between queries.
\end{definition}

For our application it is convenient to use the corresponding
mini-scheme notion, which removes the setup algorithm and the public and
secret keys.

\begin{definition}[Quantum Fire Mini-Scheme \cite{CGS:2026}]
A quantum fire mini-scheme $\mathsf{QFireMini}$ consists of the following
efficient quantum algorithms.

\begin{itemize}
    \item $\mathsf{QFireMini.Spark}(1^\lambda)$:
    Output a serial number $sn$ and a flame state in
    $R_{\mathrm{flame}}$.

    \item $\mathsf{QFireMini.Clone}(sn,R_{\mathrm{flame}})$:
    Output $R_{\mathrm{flame}}$ together with a new flame register
    $R'_{\mathrm{flame}}$.

    \item $\mathsf{QFireMini.Ver}(sn,R)$:
    Output $\mathsf{TRUE}$ or $\mathsf{FALSE}$.
\end{itemize}

Verification correctness, cloning correctness, and strong cloning
correctness are defined analogously to the full scheme.
\end{definition}

\begin{definition}[Untelegraphability of a Quantum Fire Mini-Scheme]
Let $\mathcal A=(\mathcal A_1,\mathcal A_2)$.
Consider the following game.

\begin{enumerate}
    \item The challenger samples
    \[
        sn,R_{\mathrm{flame}}
        \leftarrow\mathsf{QFireMini.Spark}(1^\lambda)
    \]
    and sends $(sn,R_{\mathrm{flame}})$ to $\mathcal A_1$.

    \item $\mathcal A_1$ outputs a polynomial-size classical string $L$.

    \item $\mathcal A_2(L)$ outputs a quantum register $R'$.

    \item The challenger outputs
    \[
        \mathsf{QFireMini.Ver}(sn,R').
    \]
\end{enumerate}

The mini-scheme is \emph{untelegraphable} if every QPT adversary
(or every query-bounded adversary in an unconditional oracle
construction) wins with probability at most
$\operatorname{negl}(\lambda)$.
\end{definition}

\paragraph{From mini-schemes to public-key schemes.}
Use a digital signature scheme secure in the same adversary model.
Setup generates $(vk,sk)$. Spark generates $(s,R)$ with the mini-scheme
and outputs $((s,\sigma),R)$, where $\sigma\leftarrow\DS.\Sign(sk,s)$.
Verification checks $\DS.\Ver(vk,s,\sigma)$ and then the mini-scheme
verifier; cloning uses the mini-scheme cloner. An unsigned new serial
breaks signature security. For a previously issued serial, embed the
mini-scheme challenge into a uniformly chosen issuance; the other issued
flames can be generated independently, giving only a polynomial loss.
This is the standard compiler of~\cite{CGS:2026}.

\newpage% ===== APPENDIX: FIRE SECURITY PROOF =====
\section{Missing Security Proof}
\label{app:missing-security}
We supply the quantum-fire security proof deferred from the main text.
Quantum-money and honest-generation signature-token security remain
conjectural; standard OSS security fails for the stated generation
procedure.

\paragraph{Inputs from the preceding code construction.}
We use the following already established guarantees. For $H\in\Good$,
the advice state is accepted with overwhelming probability and our cloner
outputs a state negligibly close to
$(\AdvState{H}\bra{\mathrm{adv}_H})^{\otimes2}$.
The base verifier samples $z\leftarrow\mathcal Z_\lambda$, so for
$E\subseteq F\times\Sigma^n$ with $|F|\le|\mathcal Z_\lambda|/3$,
its acceptance probability is at most $1/3$ for every input state.
Witness-preserving amplification~\cite{MW05} gives a verifier
$\widetilde V_M$ with soundness error $\delta_\lambda\le2^{-\lambda}$
for all such restricted oracles. We also use the list-diversity property
from the code analysis: for $L$ distinct codewords $c^{(1)},\ldots,c^{(L)}$,
with $L\le2^\lambda$, the sets
$S_i=\{c_i^{(j)}:j\in[L]\}$ satisfy
$\sum_i|S_i|\ge nL/2$ for sufficiently large $\lambda$.
This is the list-recovery consequence used in the no-query guessing
bound~\cite{BHV:2026}.

\begin{lemma}
\label{lem:fire-synthesis}
Let $B^{M[H],\mathcal R}(w)$ output a quantum register, where
$w\in\{0,1\}^t$, and where the auxiliary oracle $\mathcal R$ is
independent of $H$. Let $Q\ge1$ bound the target-oracle queries of
$B$ followed by $\widetilde V_M$. For every admissible $L$ in
Theorem~\ref{thm:code-fire},
\begin{align*}
 &\mathbb{E}_{H,\mathcal R}\!\left[
 \mathbf1_{H\in\Good}\max_{w\in\{0,1\}^t}
 \operationname{Pr}\!\left[
 \widetilde V_M^{M[H]}
       (B^{M[H],\mathcal R}(w))=1\right]\right]\\
 &\hspace{20mm}\le\delta_\lambda+
 2^{t/(2L)}\sqrt8\,Q(1-\lambda^{-4})^{n/4}.
\end{align*}
The maximum may depend on both $H$ and $\mathcal R$.
\end{lemma}
\begin{proof}
Fix $\mathcal R=r$; it can be incorporated into the inter-query operations
because it is independent of $H$. Denote the composed algorithm by $T_r$
and let
$p_{H,r}=\max_w\operationname{Pr}[T_r^{M[H]}(w)=1]$.
We construct a no-query guesser that simulates restricted target oracles.

Choose $w$ uniformly from $\{0,1\}^t$. Maintain a database $\Delta$ of
distinct pairs $(z,c)$ and a table of guessed bits $h(i,a)$, initially
empty. At each of $L$ rounds, simulate $T_r(w)$ with an oracle that returns
$h(i,a)$ when $(\bin(i,a),c)\in\Delta$, and $\bot$ otherwise.
Choose one of its $Q$ target-query positions uniformly, measure the
query input $(i,a,c)$, and insert $(\bin(i,a),c)$ into $\Delta$.
If $(i,a)$ has not previously been assigned a bit, guess its bit uniformly;
otherwise retain its previously guessed value. Restart each simulation
from the same classical input $w$. Pad with dummy queries if necessary.
This guesser makes no queries to $H$ or $M[H]$.

Fix $H\in\Good$ and a maximizing witness $w=w_{H,r}$. Condition on
all previous rounds having extracted distinct valid pairs and having
guessed the associated bits correctly. The simulated oracle is then
exactly $M[H,\Delta]$. Its challenge support has size at most
$|\Delta|<L\le|\mathcal Z_\lambda|/3$, so $T_r$ accepts under this
oracle with probability at most $\delta_\lambda$, regardless of the
state synthesized by its first part.

Let $\mu_j$ be the query weight, in the restricted-oracle execution, on
inputs where $M[H]$ and $M[H,\Delta]$ differ. The quantum hybrid bound
gives
\[
 (p_{H,r}-\delta_\lambda)_+
 \le2\sum_{j=1}^{Q}\sqrt{\mu_j},\qquad
 \frac1Q\sum_{j=1}^{Q}\mu_j
 \ge\frac{(p_{H,r}-\delta_\lambda)_+^2}{4Q^2},
\]
where $x_+=\max\{x,0\}$. Thus measurement at a uniformly chosen
query extracts a new valid pair with at least the latter probability.
A fresh hash-bit guess is correct with probability $1/2$; an already
correctly guessed bit introduces no loss. Each round therefore succeeds
with probability at least
$(p_{H,r}-\delta_\lambda)_+^2/(8Q^2)$.
Guessing $w_{H,r}$ initially costs $2^{-t}$, so the guesser's overall
success probability is at least
\[
 2^{-t}(8Q^2)^{-L}
 \mathbb{E}_H\!\left[
 \mathbf1_{H\in\Good}(p_{H,r}-\delta_\lambda)_+^{2L}\right].
\]

Conversely, fix any output of this no-query guesser. If all its $L$
pairs are distinct and valid, their codewords must be distinct, since
$H(c)$ determines the challenge. The list-diversity property implies
at least $nL/2$ distinct hash inputs $(i,a)$ among these codewords.
Each corresponding prescribed output bit has probability at most
$1-\lambda^{-4}$ under the independent biased hashes. Consequently,
the guesser's success probability is at most
$\beta_\lambda^L$, where
$\beta_\lambda=(1-\lambda^{-4})^{n/2}$.
It follows that
\[
 \mathbb{E}_H\!\left[
 \mathbf1_{H\in\Good}(p_{H,r}-\delta_\lambda)_+^{2L}\right]
 \le2^t(8Q^2\beta_\lambda)^L.
\]
Apply the power-mean inequality to
$X=\mathbf1_{H\in\Good}(p_{H,r}-\delta_\lambda)_+$ to obtain
\[
 \mathbb{E}_H\left[\mathbf1_{H\in\Good}p_{H,r}\right]
 \le\delta_\lambda+\mathbb{E}_H[X]
 \le\delta_\lambda+
       2^{t/(2L)}\sqrt8\,Q\sqrt{\beta_\lambda}.
\]
Finally average over $r$. This proves the lemma, including its
oracle-dependent choice of the best classical transcript.
\end{proof}

\begin{lemma}[Permutation locking]
\label{lem:fire-locking}
In the target-serial fire game, replace the public hash row at
$u=P^{-1}(s)$ by an independent hash family, while retaining the original
target flame and locked oracle. An adversary with $Q_H$ public-hash
queries and $Q_P$ forward-permutation queries changes its winning
probability by at most
\[
 2Q_H\sqrt{\operatorname{Inv}_\lambda(Q_P+2Q_H+1)}
 =O\!\left(\frac{Q_H(Q_P+2Q_H+2)}{2^{\lambda/2}}\right).
\]
This statement applies to the particular correlated global oracles
defined here, not to arbitrary additional oracles that could reveal $u$.
\end{lemma}
\begin{proof}
An equivalent sampling description first chooses a uniform serial $s$,
a target hash $H^*$, and independent hashes $K_y$ for $y\ne s$.
All are independent of $P$. The locked target oracle is $M[H^*]$,
and the other locked oracles are $M[K_y]$.
In the real experiment the public hash row at $v$ is $K_{P(v)}$, with
$K_s=H^*$. In the replaced experiment set $K_s=\widehat H$ for an
independent $\widehat H$, only for the public-hash interface. The flame
and target locked interface still use $H^*$.

Every query in the replaced experiment can be simulated with access to
forward $P$ and the independently sampled tables. A public-hash query
uses two evaluations of $P$, to compute and uncompute $P(v)$;
locked-oracle queries use their tables directly. Hence the state before
any public-hash query can be produced using at most $Q_P+2Q_H$ forward
queries. Measuring its row register and checking $P(v)=s$ would invert
$P$ with the query weight $\nu_j$ on $v=P^{-1}(s)$ as its success
probability. All simulation tables and the initial flame are independent
of $P$, so
$\nu_j\le\operatorname{Inv}_\lambda(Q_P+2Q_H+1)$.

The two experiments differ only on that row. A telescoping hybrid
evaluated on the replaced-experiment query states bounds their final
trace distance by $2\sum_{j=1}^{Q_H}\sqrt{\nu_j}$.
Measurements, the classical-only transfer between stages, and final
verification are channels and cannot increase trace distance.
The stated explicit bound follows. The asymptotic estimate uses
$\operatorname{Inv}_\lambda(T)=O((T+1)^2/2^\lambda)$~\cite{Amb:2000};
the hybrid step is the standard query argument~\cite{BBBV:1997}.
\end{proof}

\begin{proof}[Proof of Theorem~\ref{thm:code-fire}]
Conditioned on $H_u\in\Good$, each coordinate's one-preimage density
is at least $1/(2\lambda^4)$ and its zero-preimage density at least $1/2$.
Using $\lambda^6$ independent samples per coordinate, the probability
of failing to obtain both outcomes is at most
$2n\exp(-\lambda^2/2)$. Successful preparation gives the required
uniform preimage states exactly. Verification completeness and strong
cloning correctness follow from the earlier verifier and cloner
guarantees. Polynomial recursive cloning preserves correctness by the
channel hybrid bound.

For security, replace honest generation by the ideal target state
$\AdvState{H^*}$ on good hashes, assigning any fixed state on bad hashes.
This changes the game by at most $b_\lambda+\eta_\lambda$.
Apply Lemma~\ref{lem:fire-locking} to replace the public target hash row.
In this hybrid, only the target locked oracle $M[H^*]$ available to
$A_2$ depends on $H^*$. All public hash rows, $P$, and the other locked
oracles together form an auxiliary oracle $\mathcal R$ whose distribution
is independent of $H^*$. This independence follows from the serial-indexed
sampling description in the locking lemma, even though the auxiliary
components are mutually correlated.

The first stage may produce a transcript whose distribution depends on
both $H^*$ and $\mathcal R$. For each fixed pair, its conditional winning
probability is at most the maximum over all $t$-bit transcripts. The
second stage has no quantum side information from the first stage.
Thus, on good target hashes, its average success is bounded by
Lemma~\ref{lem:fire-synthesis}, with $B=A_2$. Combining the generation,
locking, and synthesis bounds gives~\eqref{eq:fire-winning-bound}.

For polynomial $t$, the choice $L=\max\{1,t\}$ is admissible for all
sufficiently large $\lambda$. Since $n\ge\lambda^5$,
$(1-\lambda^{-4})^{n/4}\le\exp(-\lambda/4)$ and
$2^{t/(2L)}\sqrt8\le4$. The simplified bound follows, and every term
is negligible for polynomial query budgets. This proves noninteractive
untelegraphability.
\end{proof}

\fi

\end{document}